\documentclass[11pt]{article}
\usepackage[colorlinks]{hyperref}
\hypersetup{linkcolor=cyan,filecolor=cyan,citecolor=cyan,urlcolor=cyan}
\usepackage{xspace}
\usepackage{amsmath,amssymb,bbm,amsthm}
\usepackage{fullpage}
\usepackage{boxedminipage}
\usepackage[boxed]{algorithm}
\usepackage{epigraph}
\usepackage[sc]{mathpazo}
\usepackage{amsmath}
\usepackage{mathtools}
\usepackage{framed}
\usepackage[framemethod=tikz]{mdframed}
\usepackage{titlesec}
\usepackage{lipsum}%
\usepackage{cleveref,aliascnt}
\usepackage{tikz}
\usepackage{wrapfig}
\usepackage{enumitem}

\usetikzlibrary{shapes.geometric}
\usetikzlibrary{arrows}
\usetikzlibrary{arrows.meta}
\usetikzlibrary{patterns}
\usetikzlibrary{shapes.misc}
\usepackage{thm-restate,color,xcolor}

\newcommand{\dist}{\mbox{\rm dist}}

\newtheorem{theorem}{Theorem}[section]
\newtheorem{lemma}[theorem]{Lemma}
\newtheorem{meta-theorem}[theorem]{Meta-Theorem}
\newtheorem{claim}[theorem]{Claim}

\newtheorem{corollary}[theorem]{Corollary}
\newtheorem{proposition}[theorem]{Proposition}

\newtheorem{definition}[theorem]{Definition}
\newtheorem{question}[theorem]{Question}
\newtheorem{fact}[theorem]{Fact}

\newcommand{\poly}{\operatorname{\text{{\rm poly}}}}

\newcommand{\ID}{\operatorname{ID}}
\newcommand{\EID}{\operatorname{EID}}
\newcommand{\UID}{\operatorname{UID}}

\newcommand{\Sketch}{\mathsf{Sketch}\xspace}
\newcommand{\XOR}{\mathsf{XOR}\xspace}

\def\port{\mbox{\tt port}}

\newcommand{\FTConnLabel}{\mathsf{ConnLabel}\xspace}
\newcommand{\FTDistLabel}{\mathsf{DistLabel}\xspace}
\newcommand{\TreeCover}{\mathsf{TC}\xspace}
\newcommand{\LCALabel}{\mathsf{ANC}\xspace}

\usepackage{geometry}
\newcommand{\remove}[1]{}

\newcommand{\mtodo}[1]{\textcolor{blue}{[Michal: #1]}}
\newcommand{\mertodo}[1]{\textcolor{violet}{[Merav: #1]}}

\DeclareMathOperator*{\Moplus}{\text{\raisebox{0.25ex}{\scalebox{0.8}{$\bigoplus$}}}}

\begin{document}
\date{}

\title{Fault-Tolerant Labeling and Compact Routing Schemes}
\author{
 Michal Dory\\
  \small ETH Zurich\\
  \small michal.dory@inf.ethz.ch
  			\thanks{Supported in part by the Swiss National Foundation (project grant $200021\_184735$).} 
\and
Merav Parter\\
        \small Weizmann Institute \\
        \small merav.parter@weizmann.ac.il 
			\thanks{Supported by the European Research Council (ERC) No. 949083, and by the Israeli Science Foundation (ISF) No. 2084/18.}
}
\maketitle

\begin{abstract}
The paper presents fault-tolerant (FT) labeling schemes for general graphs, as well as, improved FT routing schemes. 
For a given $n$-vertex graph $G$ and a bound $f$ on the number of faults, an $f$-FT connectivity labeling scheme is a distributed data structure that assigns each of the graph edges and vertices a short label, such that given the labels of a vertex pair $s$ and $t$, and the labels of at most $f$ failing edges $F$, one can determine if $s$ and $t$ are connected in $G \setminus F$. The primary complexity measure is the length of the individual labels. Since their introduction by [Courcelle, Twigg, STACS '07], compact FT labeling schemes have been devised only for a limited collection of graph families. In this work, we fill in this gap by proposing two (independent) FT connectivity labeling schemes for general graphs, with a nearly optimal label length. This serves the basis for providing also FT approximate distance labeling schemes, and ultimately also routing schemes. Our main results for an $n$-vertex graph and a fault bound $f$ are:
\begin{itemize}
\item There is a randomized FT connectivity labeling scheme with a label length of $O(f+\log n)$ bits, hence optimal for $f=O(\log n)$. This scheme is based on the notion of cycle space sampling [Pritchard, Thurimella, TALG '11].

\item There is a randomized FT connectivity labeling scheme with a label length of $O(\log^3 n)$ bits (independent of the number of faults $f$). This scheme is based on the notion of linear sketches of [Ahn et al., SODA '12]. 

\item For a given stretch parameter $k\geq 1$, there is a randomized routing scheme that routes a message from $s$ to $t$ in the presence of a set $F$ of faulty edges (unknown to $s$) over a path of length $O(|F|^2 k)\cdot \dist_{G\setminus F}(s,t)$. The routing labels have $\widetilde{O}(f)$ bits, the messages have $\widetilde{O}(f^3)$ bits, and each routing table has only $\widetilde{O}(f^3 n^{1/k})$ bits\footnote{Throughout the paper, we use the notation $\widetilde{O}$ to hide poly-logarithmic in $n$ terms.}. The results also holds for weighted graphs with positive polynomial weights.
\end{itemize}
This significantly improves over the state-of-the-art bounds by [Chechik, ICALP '11], providing the first scheme with sub-linear FT labeling and routing schemes for general graphs.

%

%
%

\end{abstract}

\newpage

\tableofcontents

\newpage

\section{Introduction}
Distributed graph representation is concerned with augmenting each vertex (and possibly also edges) with useful and low-space information in order to efficiently address various graph queries in a distributed manner. As the
vertices and edges of the network may occasionally fail or malfunction, it is desirable
to make these representations robust against failures. In this paper, we provide new constructions of succinct \emph{labeled-based distributed data structures} that can handle connectivity, distance queries and routing in the presence of edge failures. 

Connectivity labels are short names attached to each vertex in the $n$-vertex input graph $G$, such that given the labels of a pair of vertices $s$ and $t$ (and no any other information), it is possible to deduce if $s$ and $t$ are connected in $G$. The primary complexity measure of the  labeling scheme is the label length (maximum length of a label). In general, labels can be viewed as the \emph{logical} names of the vertices \cite{kannan1992implicat,peleg2005informative}, as they are considerably more informative than the physical names that usually correspond to arbitrary $O(\log n)$-bit identifiers. For example, in routing applications the label of the vertex is treated as its ``address". It is quite immediate to provide connectivity labeling schemes of logarithmic length. Over the years, these labels have served the basis for devising also approximate distance labels, and compact routing schemes, which are arguably the \emph{grand finale} of the distributed representation schemes. 

Our goal in this paper is to provide \emph{fault-tolerant} analogs for the above mentioned schemes, while paying a small overhead in terms of space and other complexity aspects. Several notions of fault-tolerant labeling and routing schemes have been addressed in the literature; starting with the earlier introduction of FT routing schemes by Dolev \cite{dolev1984new}, to the more recent formulations of forbidden-set labeling and routing schemes by Courcelle et al. \cite{courcelle2007forbidden,CourcelleT07}. Despite much activity revolving these topics, FT labeling and routing schemes with \emph{sub-linear} space are currently known only for a limited collection of graph families. We next elaborate more on the state-of-the-art affairs, and our main objectives.
%
%

\paragraph{Fault-Tolerant Connectivity and Distance Labeling.} 
FT connectivity labeling schemes, also known in the literature as \emph{forbidden-set} labeling \cite{CourcelleT07}, assign labels to the vertices and the edges of the graph such that given the labels of a vertex pair $s,t$, and the labels of the faulty edges $F$, one can determine if $s$ and $t$ are connected in $G \setminus F$. 
Since their introduction, efficient FT labeling schemes have been devised only for a restricted collection of graph families such as graphs with bounded tree-width and planar graphs \cite{CourcelleT07,AbrahamCGP16}. In the lack of any FT connectivity labeling schemes for general graphs with sub-linear label length (for any $f\geq 2$ faults\footnote{While there is no \emph{explicit} construction of FT labeling for general graphs, for $f=1$, the centralized distance sensitivity oracle of \cite{khanna2010approximate} might be modified to provide approximate distance labels against a single fault.}), we ask:

\begin{question}\label{q:label}
Is it possible to design FT connectivity labeling scheme resilient to at most $f$ edge faults, for general graphs with label length of $\poly(\log n)$ bits, or even $\poly(\log n,f)$ bits? 
\end{question}

FT connectivity labels are also closely related to \emph{connectivity sensitivity oracles} \cite{patrascu2007planning}, which are low-space centralized data-structures that handle efficiently $\langle s, t, F \rangle$ connectivity queries using $S(n)$ space. Our main goal is in providing a \emph{distributed} variant of such constructions, e.g., where each vertex or edge in the graph ``holds" only $S(n)/n$ bits of information, such that an $\langle s,t, F \rangle$ query can be addressed using only the information stored by $s,t$ and $F$. 

An important step towards designing FT compact routing schemes involves the computation of \emph{FT approximate distance labels}. In this setting, given the labels of $s,t$ and the faulty edges $F$, it is required to report an approximation for the $s$-$t$ shortest path distance in $G \setminus F$. 
FT approximate distance labels can be viewed as the distributed analog of $f$-FT \emph{distance sensitivity oracles} \cite{khanna2010approximate,WeimannY10}. 
These are global succinct data-structures that given an $\langle s,t,F \rangle$ query report fast an estimate for the approximate $s$-$t$ distance in $G \setminus F$. Our goal is to provide FT approximate labeling schemes that match the state-of-the-art space vs. stretch tradeoff of the centralized data structures. 

\paragraph{Fault-Tolerant Routing.} A desirable requirement in most communication networks is to provide efficient routing protocols in the presence of faults. Specifically, an $f$-FT routing protocol is a distributed algorithm that, for any set of at most $f$ faulty edges $F$, allows a vertex $s$ to route a message to a destination vertex $t$ along an approximate $s$-$t$ shortest path in $G \setminus F$ (without knowing $F$ in advance). The routing scheme consists of two algorithms: (i) a preprocessing algorithm which computes (succinct) routing tables and labels for each vertex in the graph; and (ii) a routing algorithm that given the received message and the routing table of vertex $v$ determines the next-hop (specified as a port number) on the $v$-$t$ (approximate) shortest path in $G \setminus F$. The efficiency of the scheme is determined by the tradeoff between the \emph{stretch} (i.e., the ratio between the weighted length of the $s$-$t$ route in $G\setminus F$ to the corresponding shortest path distance) and the \emph{space} of the routing tables, labels and messages. 
%
%
%
While the stretch vs. space tradeoff of routing schemes is fully understood in the non-faulty setting, the corresponding bounds in the FT setting are still far from optimal. 
So far, in all the prior schemes, the space of the individual routing tables could be linear in the worst case, even when allowing a large stretch bound. This is in strike contrast to the standard (non-faulty) compact routing schemes, e.g., by Thorup and Zwick \cite{thorup2001compact}, which provide each vertex a table of $\widetilde{O}(n^{1/k})$ bits, while guaranteeing a route stretch of $2k-1$. The current large gap in the quality of FT routing schemes compared to their non-faulty counterparts leads to the following question.

\begin{question}\label{q:route}
Is it possible to design $f$-fault-tolerant compact routing scheme for general graphs with \emph{sub-linear} table size and with a sub-logarithmic stretch?
\end{question}



\subsection{Our Results}
We provide space-efficient labeling and routing schemes for any $n$-vertex graph. Our schemes are \emph{randomized} and provide a high probability guarantee\footnote{As standard, we use the term high-probability to indicate success guarantee of $1-1/n^c$ for any given constant $c>1$.} for any given triplet $\langle s, t, F\rangle$. In other words, the schemes can faithfully support polynomially many queries\footnote{The same type of guarantee is provided in the centralized sensitivity oracles, e.g., of \cite{DuanConnectivitySODA17}. Providing a high probability guarantee over all possible triplets is possible upon increasing the space bound by a factor of $f$ (largest number of faults supported).}. 

Our first key result presents two independent schemes for FT connectivity labels. These are the first FT connectivity labels for \emph{general graphs}.
These two constructions yield the following theorem, addressing Question \ref{q:label}: 

\begin{theorem}\label{thm:conn-labels}[FT Connectivity Labeling Schemes, Informal]
For any $n$-vertex graph and a bound $f$ on the number of edge faults, there is a \emph{randomized} $f$-FT connectivity labeling scheme with label length of $O(\min\{f+\log n, \log^3 n\})$ bits. The labels are computed in $\widetilde{O}(m)$ time, and the decoding algorithm takes $\poly(f,\log n)$ time. 
\end{theorem}

By the tightness of the label length of fault-free connectivity labels, 
our scheme is optimal for $f=O(\log n)$. Moreover, the label length is nearly-optimal for any $f$. Our actual scheme provides more information then merely a single bit (connected or not connected). Specifically, we augment the connectivity labels with additional information so that the decoding algorithm, given the labels of $s,t$ and $F$, can also output a succinct description of an $s$-$t$ path in $G\setminus F$ (if such a path exists). This succinct path representation finds applications in the context of our FT routing schemes. 
%

We next consider the task of reporting also approximate $s$-$t$ distances in $G\setminus F$ using the labels of $s,t$ and $F$. We employ the reduction of Chechik et al. \cite{chechik2012f} to convert the FT connectivity labels into FT approximate distance labels, providing nearly the same space vs. stretch tradeoff as in the centralized data-structures of \cite{chechik2012f}. Specifically, we show:

\begin{theorem}\label{thm:dist-labels}[FT Approximate Distance Labeling Schemes]
For any $n$-vertex (possibly weighted) graph, a bound $f$ on the number of edge faults, and a stretch parameter $k$, there is a randomized $f$-FT approximate distance labeling scheme with label length of $O(k \cdot n^{1/k}\cdot \log (nW) \cdot \log^3 n)$. Given the labels of $s,t$ and $F$ the scheme returns a distance estimate 
$$\dist_{G\setminus F}(s,t)\leq \delta(s,t,F) \leq (8k-2)(|F|+1)\dist_{G\setminus F}(s,t)~.$$
\end{theorem}

For the purpose of routing, we exploit the extra information provided by our connectivity labels, in order to output, in addition to the distance estimate $\delta(s,t,F)$, also a succinct description of the approximate $s$-$t$ 
shortest path in $G \setminus F$. Our second key result provides FT compact routing schemes, with an almost optimal tradeoff between the space and stretch, for constant number of faults $f$. We answer Question \ref{q:route} by showing:

\begin{theorem}\label{thm:routing}[FT Compact Routing]
For every integers $k,f$, there exists an $f$-sensitive compact routing scheme that given a message $M$ at the source vertex $s$ and the routing label of the destination $t$, in the presence of at most $f$ faulty edges $F$ (unknown to $s$) routes $M$ from $s$ to $t$ in a distributed manner over a path of length at most $32k (|F|+1)^2\cdot \dist_{G \setminus F}(s,t)$. The routing labels have $\widetilde{O}(f)$ bits, the table size of each vertex is $\widetilde{O}(f^3 \cdot n^{1/k} \log(nW))$, the header size (also known as message size) is bounded by $\widetilde{O}(f^3)$ bits. 
\end{theorem}
This improves over the state-of-the-art construction of Chechik \cite{chechik2011fault} that 
obtained routing schemes with stretch of $O(f^2(f+\log^2 n)k)$ and tables of size $O(\deg(v)n^{1/k}\log{(nW)})$ for every vertex $v$. We note that the construction of Chechik \cite{chechik2011fault} has a bounded global space of $\widetilde{O}(n^{1+1/k} \log{(nW)})$, but the individual tables might have even super-linear space (e.g., when $k=O(1)$ and $\deg(v)=O(n)$).  For the special case of $f=2$, Chechik et al. \cite{ChechikLPR10,chechik2012f} provide a stretch bound of $O(k)$, and total space of $\widetilde{O}(n^{1+1/k} \log{(nW)})$, where the space of each table is bounded by $O(\deg(v)n^{1/k})$, thus super-linear in the worst case. Our scheme provides an improved bound on the individual tables, nearly matching the fault-free constructions for $f=O(1)$. We also show an improved scheme if one only aims to optimize for the global space, rather than optimizing for the largest table size for a vertex. For comparison of our results to prior work see Table \ref{table_routing}. 

\begin{table}[h!]
\centering
\begin{tabular}{ |p{4.2cm}|p{4cm}|p{6cm}|p{0.5cm}|}
 \hline	
 \multicolumn{4}{|c|}{Constructions of Fault-Tolerant Routing Schemes}\\
 \hline
  Reference & Stretch & Table Size & $|F|$ \\
 \hline
  Rajan \cite{rajan2012space} & $O(k^2)$ & $\widetilde{O}(k \deg(v)+ n^{1/k})$ per vertex & 1\\
  Chechik et al. \cite{chechik2012f} & $O(k)$ & $\widetilde{O}(n^{1+1/k} \log(nW))$ total size & 2\\
  Chechik \cite{chechik2011fault}  & $O(|F|^2(|F|+\log^2 n)k)$ & $\widetilde{O}(n^{1+1/k} \log(nW))$ total size & $f$\\
  Chechik \cite{chechik2011fault}  & $O(|F|^2(|F|+\log^2 n)k)$ & $\widetilde{O}(\deg(v)n^{1/k} \log(nW))$ per vertex & $f$\\
  \textbf{Here} & $O(|F|^2 k)$ & $\widetilde{O}(f \cdot n^{1+1/k} \log(nW))$ total size & $f$\\
  \textbf{Here} & $O(|F|^2 k)$ & $\widetilde{O}(f^3 \cdot n^{1/k} \log(nW))$ per vertex & $f$\\
 \hline
 \end{tabular}
 \caption{Comparison between FT routing schemes with a set of failures $F$}
\label{table_routing}
\end{table}

Finally, we provide a lower bound result on the minimal stretch regardless for the \emph{space} of the routing scheme, e.g., even if all vertices store all the graph edges. 

\begin{theorem}[Stretch Lower-Bound for FT Routing]\label{thm:lb-routing}
Any FT routing randomized scheme resilient to $f$ faults induces an expected stretch of $\Omega(f)$ regardless of the size of the routing tables and labels. In particular, this holds even if each routing table contains a complete information on the graph. 
\end{theorem}

\paragraph{Open Problems.} Our work leaves several interesting open ends. One natural direction is to provide labeling and routing schemes resilient to \emph{vertex} faults. The major challenge in handling vertex failure is that even a single faulty vertex might disconnect the graph into $\Omega(n)$ disconnected components. Another interesting direction is to derandomize our constructions. Currently there are no deterministic constructions of FT labeling schemes for general graphs. Finally, it will be also important to provide FT distance approximate labeling schemes whose stretch bound is independent in the number of faults $f$. This problem is also open in the corresponding setting of approximate distance sensitive oracles.

\subsection{Our Techniques} \label{sec:techniques}

For our FT labeling schemes, we present two constructions based on different techniques. 
The first construction uses the \emph{cycle-space sampling} technique of Pritchard and Thurimella \cite{pritchard2011fast} to determine if $s$ and $t$ are disconnected by a set of failures $F$. This technique has been applied in the past mainly in the context of computing small cuts in the distributed setting. 
The second construction uses the tool of \emph{linear sketches} by Ahn et al. \cite{ahn2012analyzing} to try to find a path that connects $s$ and $t$ in $G \setminus F$. This scheme is also useful for routing.
We next give an overview of the two approaches, and the applications for routing. Throughout, we assume that the graph $G$ is originally connected, otherwise the scheme can be applied to each connected component of $G$, which can be indicated in the label of the vertex.

\paragraph{Connectivity Labels Based on Cycle Space Sampling.} The cycle space sampling technique, introduced by Pritchard and Thurimella \cite{pritchard2011fast}, allows one to detect cuts in a graph by exploiting the interesting connection between cuts and cycles in a graph. This technique was used in \cite{pritchard2011fast} to design distributed algorithms for identifying small cuts in a graph. In more details, the technique is based on the relation between \emph{induced edge cuts} and \emph{binary circulations}, defined as follows. For a subset of vertices $S$, we denote by $\delta(S)$ the set of edges with exactly one endpoint in $S$. An \emph{induced edge cut} is a set of edges of the form $\delta(S)$ for some $S$. A \emph{binary circulation} is a set of edges in which every vertex has an even degree. For example, a cycle is a binary circulation. Note that if $F$ is an induced edge cut, and $\phi$ is a cycle, the number of edges in the intersection $|F \cap \phi|$ is even, as the cycle crosses the cut even number of times. This is also true for any binary circulation $\phi$. The cycle space technique extends this observation and shows that if $\phi$ is a random binary circulation and $F \subseteq E$, then
$$Pr[|F \cap \phi| \ is \ even] = \left\{
                \begin{array}{ll}
                  1,\ if\ F\ is\ an\ induced\ edge\ cut\\
                  1/2,\ otherwise
                \end{array}
              \right. $$ 
Hence, by choosing a \emph{random} binary circulation, one can detect if a set of edges $F$ is an induced edge cut with probability $1/2$.  To increase the success probability, we can choose $b$ random binary circulations. 
Based on these ideas, \cite{pritchard2011fast} showed how to assign the edges of the graph $b$-bit labels with the following property.  See Appendix \ref{sec:cycle_space_overview} for an overview.

\begin{restatable}{lemma}{cycle} \label{cycle_space_lemma}
There is an algorithm that assigns the edges of a graph $G=(V,E)$, $b$-bit labels $\phi(e)$ such that given a subset of edges $F \subseteq E$, we have:
$$Pr[\Moplus_{e \in F} \phi(e) = 0] = \left\{
                \begin{array}{ll}
                  1,\ if\ F\ is\ an\ induced\ edge\ cut\\
                  2^{-b},\ otherwise
                \end{array}
              \right. $$ 
Where $0$ is the all-zero vector. The time complexity for assigning the labels is $O((m+n)b)$.
\end{restatable}

\noindent\textbf{The connectivity labels.} We next explain how to use this technique to build FT connectivity labels. Our goal is to assign labels to the vertices and edges of the graph, such that given the labels of two vertices $s,t$ and a set of failures $F$, we can check if $s$ and $t$ are disconnected by $F$. It is easy to show that $s$ and $t$ are disconnected by $F$ iff there is an \emph{induced edge cut} $F' \subseteq F$ that disconnects $s$ and $t$. While we can use the cycle space labels to check if a subset of edges $F' \subseteq F$ is an induced edge cut, this is still not enough to solve FT connectivity. To do so, we should check if an induced edge cut $F'$ \emph{disconnects} the vertices $s$ and $t$. To check this, we bring to our construction \emph{ancestry labels} in trees, and show that we can determine if $s$ and $t$ are in the same side of cut (induced by $F'$) based on the ancestry labels of $s,t$ and $F'$. The key observation is that a spanning tree $T$ of the graph is disconnected to at most $|F'|+1$ connected components, upon removing $F'$, where for any $e \in F'$ both its endpoints reside on two different sides of the induced edge cut defined by $F'$.
We can use this to identify which components of $T \setminus F'$ are on the same side of the induced edge cut. Moreover, we show that the ancestry labels allow us to determine the connected components of $s$ and $t$ in $T \setminus F'$. A brute-force implementation of this approach leads to a decoding time that is \emph{exponential} in $|F|$. I.e., the algorithm should check for any subset $F' \subseteq F$ if $F'$ is an induced edge cut. To overcome it, we show an efficient way to find $F' \subseteq F$ that disconnects $s$ and $t$ if exists, by translating our problem to a system of linear equations. This results in a decoding time polynomial in $|F|$ and $\log{n}$. The size of the labels is $O(f+\log{n})$, to guarantee that the cycle space labels are correct for any $F' \subseteq F$ w.h.p.      

\paragraph{Connectivity Labels Based on Graph Sketches.} We next provide some flavor of our labels based graph sketches. The length of the labels obtained in this technique is $O(\log^3{n})$ bits, which is dominated by the sketching information. A \emph{graph sketch} of a vertex $v$ is a randomized string of $\widetilde{O}(1)$ bits that compresses $v$'s edges. The linearity of these sketches allows one to infer, given the sketches of subset of vertices $S$, an outgoing cut edge $(S, V \setminus S)$. Graph sketches have numerous applications in the context of connectivity computation under various computational settings, e.g., \cite{kapron2013dynamic,kapralov2014spanners,GibbKKT15,DBLP:conf/podc/KingKT15,DBLP:conf/wdag/MashreghiK18,GhaffariP16,DuanConnectivitySODA17}. More concretely, our sketch-based labels are inspired by the centralized connectivity sensitivity oracles of Duan and Pettie \cite{DuanConnectivitySODA17}.   A common approach for deducing the graph connectivity merely from the sketches of the individual vertices is based on the well-known Boruvka algorithm \cite{Boruvka}. This algorithm works in $O(\log n)$ phases, where in each phase, from each growable component an outgoing edge is selected. All these outgoing
edges are added to the forest, while ignoring cycles. Each such phase reduces the number of
growable components by a $2$ factor, thus within $O(\log n)$ phases, a maximal forest is computed. Since this algorithm only requires the computation of outgoing edges it can simulated using $O(\log n)$ independent sketches for each of the vertices. 

Our high level approach for determining the $s$-$t$ connectivity in $G \setminus F$ mimics this above mentioned procedure. For simplicity assume that $G$ is connected and let $T$ be some spanning tree in $G$. Using ancestry labels, one can infer the components of $T \setminus F$. Moreover, by augmenting the labels with graph sketching information, one can also deduce the sketch of each component in $T \setminus F$. 
Note however that these sketches are in $G$ and therefore might encode outgoing edges that belong to $F$. To overcome this technicality, our sketching scheme allows us to cancel out the effect of the faulty edges $F$ from the sketching information. Consequently, we obtain the sketches of each $T \setminus F$ component in the surviving graph $G \setminus F$. We can then apply the Boruvka's algorithm on the components of $T \setminus F$, and infer the $s$-$t$ connectivity in $G \setminus F$. The actual implementation of this labeling scheme is somewhat more delicate. We note that some of these technicalities are for the sake of our later extension of these labels into compact routing schemes.

\paragraph{Applications for Routing Schemes.} The starting point to our routing scheme is given by our (sketch-based) labeling scheme. These labels allows one to deduce also a succinct description of an $s-t$ path in $G \setminus F$ if exists, by following the component merging procedure of the Boruvka algorithm. This description is composed of $O(f)$ path segments, where each segment $\{u,v\}$ either corresponds to an outgoing (non-tree) edge found in the algorithm using the sketch information, or to a tree path between two vertices $u$ and $v$ in the same connected component in $T \setminus F$. Given the connectivity labels of $s,t$ and $F$, we can find this description, and use it for routing. Routing across an edge $\{u,v\}$ just requires sending a message over the edge, while routing on a tree path between $u$ and $v$ can be done using a routing scheme for trees. 
While this approach allows to send a message from $s$ to $t$, there is no bound on the length of the path traversed. Additionally, this approach assumes that the set of failures $F$ is known in advance. We next explain how to overcome these issues.
\\
\noindent\textbf{Bounding the stretch.} To route messages on low-stretch paths we use the notion of \emph{tree covers}, following the approach in \cite{chechik2012f}. This approach also allows us to translate our connectivity labels to approximate distance labels as we discuss in Section \ref{sec:ft-distance}. Here, instead of applying our connectivity scheme on just one graph $G$, we apply it on many subgraphs $G_{i,j}$ of $G$ with the following properties. 
\begin{enumerate}
\item Each vertex $v$ is contained in $\widetilde{O}(n^{1/k})$ subgraphs.
\item For any $1 \leq i \leq \log(nW)$, and any vertex $v$, there is a subgraph $G_{i,i^*(v)}$ that contains all the vertices in the $2^i$-neighborhood of $v$.
\item If $v$ and $u$ are connected in the graph $G_{i,i^*(v)} \setminus F$, then there is a path between them of length at most $O(k|F| 2^i)$ in the graph $G_{i,i^*(v)} \setminus F$.\label{prop_path}
\end{enumerate}
By applying our connectivity scheme on each one of the subgraphs $G_{i,j}$, we can route a message from $s$ to $t$ on a path of stretch $O(k|F|)$. The size of the labels and routing tables of vertices is $\widetilde{O}(n^{1/k})$ as each vertex and edge participate in $\widetilde{O}(n^{1/k})$ subgraphs.
%
%
\\
\noindent\textbf{Faulty edges are unknown.} The scheme we described assumes that the routing algorithm knows the labels of $s,t$ and $F$ in advance, we next explain how to avoid this assumption. Our general approach is to work in phases, where in each phase we try to route a message from $s$ to $t$ according to the currently set of known faults. We either succeed, or learn about the label of a new faulty edge $e \in F$ and try again. The stretch of the scheme increases to $O(k|F|^2)$ because of the $|F|+1$ phases. Direct application of this approach may require large routing tables, as each vertex may need to know the labels of all edges adjacent to it, to be able to learn the labels of faulty edges found in the algorithm. To overcome it we use the following ingredients. 

First, recall that in our connectivity labeling scheme we use a spanning tree $T$. In the routing scheme, these are the trees of the tree cover. We show that it is enough for each vertex to store labels only of its adjacent \emph{tree} edges. 
%
Consequently, the total size of all routing tables can be bounded by $\widetilde{O}(fn^{1+1/k})$.\footnote{The $f$ term in the size comes from the fact we apply the connectivity labels $f+1$ times to support the $|F|+1$ phases.}
However, this alone is not enough to bound the size of individual routing tables of vertices, as the degree of a vertex in a tree may be linear. To overcome this, we show a clever way to load balance the labels' information between $v$ and its children in the tree. This results in tables of size $\widetilde{O}(f^3 n^{1/k})$ per vertex, while keeping the same stretch of the scheme. The increase in the total size of tables comes from the fact we now duplicate labels $f+1$ times, to be able to recover them in the presence of $f$ failures. 
\subsection{Additional Related Work}

\paragraph{Fault-Tolerant Labeling Schemes.} FT labels for connectivity were introduced by \cite{courcelle2007forbidden} under the term \emph{forbidden-set labeling}. Forbidden set refers to a subset $F$ of at most $f$ edges, such that given the labels of $s,t$ and $F$ one should determine if $s$ and $t$ are connected in $G \setminus F$. The forbidden edge set can be treated in this context as faulty edges\footnote{For routing, the forbidden-set scheme is slightly weaker than FT scheme as explained later.}.
Previous works study FT connectivity labels only in restricted graph families. For example, Courcelle et al. \cite{CourcelleT07} presented a labeling scheme with logarithmic label length for the families of $n$-vertex graphs with bounded clique-width, tree-width and planar graphs. For $n$-vertex graphs with doubling dimension at most $\alpha$, Abraham et al. \cite{AbrahamCGP16} designed FT labeling schemes with label length $O((1 + 1/\epsilon)^{2\alpha}\log n)$ that output $(1+\epsilon)$ approximation of the shortest path distances under faults. Recently, \cite{DBLP:journals/corr/abs-2102-07154} studied FT exact distance labels in planar graphs, and show that any directed weighted planar graph admits fault-tolerant distance labels of size $O(n^{2/3})$.

\paragraph{Connectivity and Distance Sensitivity Oracles.} 
Connectivity and distance sensitivity oracles are centralized data structures that support connectivity or distance queries in the presence of failures. 
The first construction of connectivity sensitivity oracles was given by Patrascu and Thorup \cite{patrascu2007planning} providing an $S(n)=\widetilde{O}(fn)$ space oracle that answers $\langle s,t, F \rangle$ connectivity queries in $\widetilde{O}(f)$ time. The state-of-the-art bounds of these oracles are given by Duan and Pettie \cite{DuanConnectivitySODA17}.
Chechik et al. \cite{chechik2012f} presented the first randomized construction of distance sensitivity oracle resilient to $f$ edge faults. 
Specifically, for any $n$-vertex weighted graph, stretch parameter $k$, and a fault bound $f$, they provide a data-structure with $O(f k n^{1+1/k}\log(nW))$ space, query time of $\widetilde{O}(|F|)$, and $O(f k)$ stretch, where $W$ is the weight of the heaviest edge in the graph. Their solution is based on an elegant transformation that converts the FT connectivity oracle of \cite{patrascu2007planning} into an FT approximate distance oracle.

While the main focus of this paper is in approximate distances, sensitivity oracles that report (possibly near) exact distances under faults have been studied also thoroughly in e.g., \cite{demetrescu2002oracles,bernstein2008improved,duan2009dual,WeimannY10,GrandoniW12,ChechikCFK17,van2019sensitive}. Since reporting exact distances requires linear label length already in the fault-free setting \cite{gavoille2004distance}, we focus on the approximate relaxation, where there is still hope to obtain labels of polylogarithmic length.

\paragraph{Fault-Tolerant Routing Schemes.}
The first formalization of FT routing schemes was given by the influential works of Dolev \cite{dolev1984new} and Peleg \cite{peleg1987fault}. These earlier works presented the first non-trivial solutions for general graphs supporting at most $\lambda$ faulty edges, where $\lambda$ is the edge-connectivity of the graph. Their routing labels had linear size, providing $s$-$t$ routes of possibly linear length (even in cases where the surviving $s$-$t$ path is of $O(1)$ length). In competitive FT routing schemes, it is required to provide $s$-$t$ routes of length that competes with the shortest $s$-$t$ path in $G \setminus F$, even in cases where $G \setminus F$ is not connected. Competitive FT routing schemes \cite{peleg2009good} for general graphs were given by Chechik et al. \cite{ChechikLPR10,chechik2012f} for the special case of $f\leq 2$ faults. 
Specifically, for a given stretch parameter $k$, they gave a routing scheme with a total space bound of $\widetilde{O}(n^{1+1/k})$ bits, polylogarithmic-size labels and messages, and a routing \emph{stretch} of $O(k)$. 
This scheme was extended later on for any $f$ by Chechik \cite{chechik2011fault}, at the cost of increasing the routing stretch to $O(f^2(f+\log^2 n)k)$. For a single edge failure, \cite{rajan2012space} showed a routing scheme with routing tables of size $\widetilde{O}(k \deg(v)+ n^{1/k})$ size per vertex, $O(k^2)$ stretch and $O(k+\log{n})$ size header.

\paragraph{Forbidden Set Routing.}
A more relaxed setting of FT routing scheme which has been studied in the literature is given by the \emph{forbidden set routing schemes}, introduced by Courcelle and Twigg \cite{CourcelleT07}. In that setting, it is assumed that the routing protocol knows in advance the set of faulty edges $F$. In contrast, in the FT routing setting, the failing edges are a-priori unknown to the routing algorithm,  and can only be detected upon arriving one of their endpoints. Forbidden set routing schemes have been devised to the same class of restricted graph families as obtained for the forbidden set labeling setting \cite{CourcelleT07,AbrahamCGP16,abraham2012fully}.
\section{Preliminaries}
Given a graph $G=(V,E)$, and vertex $u \in V$, let $\deg(u,G)$ be the degree of $u$ in $G$. 
Given a tree $T$ and $u, v \in T$, denote the $u$-$v$ path in $T$ by $\pi(u,v,T)$. When the tree $T$ is clear from the context, we may omit it and write $\pi(u,v)$. For a (possibly weighted) subgraph $G' \subseteq G$ and a vertex pair $s,t \in V$, let $\dist_{G'}(s,t)$ denote the length of the $s$-$t$ shortest path in $G'$. 

\paragraph{Fault-Tolerant Labeling Schemes.}
For a given graph $G$, let $\Pi: V\times V \times \mathcal{G} \to \mathbb{R}_{\geq 0}$ 
be a function defined on pairs of vertices and a subgraph $G' \subset G$, where $\mathcal{G}$ is the family of all subgraphs of $G$. For an integer parameter $f\geq 1$, an $f$-\emph{fault-tolerant labeling scheme} for a function $\Pi$ and a graph family $\mathcal{F}$ is a pair of functions $(L_{\Pi},D_{\Pi})$. The function $L_{\Pi}$ is called the \emph{labeling function}, and $D_{\Pi}$ is called the \emph{decoding function}. For every graph $G$ in the family $\mathcal{F}$, the labeling function $L_{\Pi}$ associates with each vertex $u \in V(G)$ and every edge $e \in E(G)$, a label $L_{\Pi}(u,G)$ (resp., $L_{\Pi}(e,G)$). It is then required that given the labels of any triplets $s,t, F \in V \times V \times E^f$, the decoding function $D_{\Pi}$ computes $\Pi(s,t, G \setminus F)$.  The primary complexity measure of a labeling scheme is the \emph{label length}, measured by the length (in bits) of the largest label it assigns to some vertices (or edges) in $G$ over all graphs $G \in \mathcal{F}$. An $f$-FT connectivity labeling scheme is required to output YES iff $s$ and $t$ are connected in $G \setminus F$.  In $f$-FT \emph{approximate distance labeling scheme} it is required to output an estimate for the $s$-$t$ distance in the graph $G \setminus F$. Formally, an $f$-FT labeling scheme is $q$\emph{-approximate} if the value $\delta(s,t,F)$ returned by the decoder algorithm satisfies that $\dist_{G \setminus F}(s,t)\leq \delta(s,t,F) \leq q \cdot \dist_{G \setminus F}(s,t)$.  Throughout the paper we provide randomized labeling schemes which provide a high probability guarantee of correctness for any fixed triplet $\langle s,t, F \rangle$.

\paragraph{Fault-Tolerant Routing Schemes.} In the setting of FT routing scheme, one is given a pair of source $s$ and destination $t$ as well as $F$ edge faults, which are initially unknown to $s$. The routing scheme consists of \emph{preprocessing} and \emph{routing} algorithms. The preprocessing algorithm defines labels $L(u)$ to each of the vertices $u$, and a header $H(M)$ to the designated message $M$. In addition, it defines for every vertex $u$ a routing table $R(u)$. The labels and headers are usually required to be short, i.e., of poly-logarithmic bits. 
The routing procedure determines at each vertex $u$ the port-number on which $u$ should send the messages it receives. The computation of the next-hop is done by considering the header of the message $H(M)$, the label of the source and destination $L(s)$ and $L(t)$ and the routing table $R(u)$. The routing procedure at vertex $u$ might also edit the header of the message $H(M)$. The failing edges are not known in advance and can only be revealed by reaching (throughout the message routing) one of their endpoints. The \emph{space} of the scheme is determined based on maximal length of message headers, labels and the individual routing tables. The stretch of the scheme is measured by the ratio between the length of the path traversed until the message arrived its destination and the length of the shortest $s$-$t$ path in $G \setminus F$. In the more relaxed setting of \emph{forbidden-set routing schemes} the failing edges are given as input to the routing algorithm.


\section{Fault-Tolerant (FT) Connectivity Labels}


We next discuss two labeling schemes for connectivity that are based on two different approaches. The first one uses the \emph{cycle space sampling} technique to try to find cuts that disconnect $s$ and $t$. The second one uses \emph{graph sketches} to try to find a path that connects $s$ and $t$. Since the second approach allows to find a path between $s$ and $t$ if exists, it is also useful later for routing. In terms of label size, the first approach gives labels of size $O(f + \log{n})$, which is near-optimal if the number of failures is $f=O(\log{n})$. On the other hand, the second scheme gives labels of size $O(\log^3{n})$, which is better when the number of failures is large.
We next discuss the labeling schemes. During this section, we assume that the input graph $G$ is connected. If not, we can add to the label of each vertex and edge the id of their connected component in $G$, and apply the labeling scheme to each one of the connected components separately. 

\subsection{Connectivity Labels Based on Cycle Space Sampling}

\subsubsection{The Labeling Algorithm}

Our labels are composed of two ingredients, that we review next.

\paragraph{Cycle Space Labels.}
The cycle space sampling technique, introduced in \cite{pritchard2011fast}, allows to give the edges of a graph short labels that allow to detect cuts in the graph. For a set of vertices $S$, $\delta(S)$ is the set of edges with exactly one endpoint in $S$. A subset of edges $F$ is called an \emph{induced edge cut} if $F = \delta(S)$ for some $S$.
The following is shown in \cite{pritchard2011fast} (see Corollary 2.9). 

 \cycle*

\remove{
\begin{restatable}{lemma}{cycle} \label{cycle_space_lemma}
There is an algorithm that assigns the edges of a graph $G=(V,E)$ $b$-bit labels $\phi(e)$ such that given a subset of edges $F \subseteq E$, we have:
$$Pr[\Moplus_{e \in F} \phi(e) = 0] = \left\{
                \begin{array}{ll}
                  1,\ if\ F\ is\ an\ induced\ edge\ cut\\
                  2^{-b},\ otherwise
                \end{array}
              \right. $$ 
Where $0$ is the all-zero vector. The time complexity for assigning the labels is $O((m+n)b)$.
\end{restatable}
}

For an overview of the technique, see Appendix \ref{sec:cycle_space_overview}. 
In our algorithm, given a subset of edges $F$ of size at most $f$, we want to be able to check for any subset $F' \subseteq F$ if $F'$ is an induced edge cut. To support all these $2^f$ queries w.h.p we choose $b=f+ c \log{n}$ for a constant $c$. This guarantees that the probability of error is at most $\frac{2^f}{2^{f+c\log{n}}}=\frac{1}{n^c}$. This will guarantee that given a query $\langle s,t,F \rangle$, our algorithm answers correctly w.h.p. We remark that if we increase the size of labels to $O(f \log{n})$ we can get an algorithm that is correct for \emph{all} queries w.h.p. 
The reason is that we can then check for any subset of edges $F$ of size at most $f$ if $F$ is an induced edge cut. As the number of subsets of size at most $f$ is bounded by $O(n^f)$, we get that the labels are correct for all such subsets w.h.p.


\paragraph{Ancestry Labels.} Our second ingredient are ancestry labels for trees.
To use them, we first fix a spanning tree $T$ of the graph rooted at $r$. The goal is to assign vertices short labels, such that given the labels of $u$ and $v$, we can infer if $u$ is an ancestor of $v$ in $T$. A simple labeling scheme based on a DFS scan solves the problem with labels of size $2 \lceil \log{n} \rceil$ per vertex \cite{kannan1992implicat}, the time for assigning the labels is $O(n)$ for the DFS scan of the tree. Labeling schemes with improved label size appear in \cite{abiteboul2006compact,alstrup2002improved,fraigniaud2010compact,fraigniaud2010optimal}.

\begin{lemma} \label{anc_labels}
For every tree $T$, there is an algorithm that assigns the vertices $u$ of the tree labels $\LCALabel_T(u)$ of $O(\log{n})$ bits, such that given the labels of $u$ and $v$ we can infer if $u$ is an ancestor of $v$ in $T$ in $O(1)$ time. The time for assigning the labels is $O(n)$. 
\end{lemma}

\paragraph{The Final Labels.}
Our final labels contain the following ingredients:
\begin{enumerate}
\item The label of the edge $e=(u,v)$ is composed of $(\phi(e),\LCALabel_T(u),\LCALabel_T(v),j)$, where $j$ is a bit indicating if $e$ is a tree edge in $T$. In total, the label size is $O(f + \log{n})$.
\item The label of a vertex $v$ is its ancestry label $\LCALabel_T(v)$ of size $O(\log{n})$ bits.
\end{enumerate}

As discussed, the time for assigning the labels is $O((m+n)b)=\widetilde{O}((m+n)f)$, as $b=f+c\log{n}$.
We next explain how we use these labels to check FT connectivity.

\subsubsection{The Decoding Algorithm}

We next discuss several observations that allow us to check if $s$ and $t$ are disconnected by $F$.

\begin{claim} \label{obs_induced}
The vertices $s$ and $t$ are disconnected by $F$ if an only if there is an induced edge cut $F' \subseteq F$ that disconnects $s$ and $t$.
\end{claim}

\begin{proof}
First, if $F' \subseteq F$ disconnects $s$ and $t$, then clearly $F$ disconnects $s$ and $t$.
On the other hand, if $s$ and $t$ are disconnected by $F$, let $F' \subseteq F$ be a minimal set of edges whose removal disconnects $s$ and $t$. We show that $F'$ is an induced edge cut. Let $V_s$ be the vertices in the connected component of $s$ in $G \setminus F'$. We show that all edges in $F'$ are between $V_s$ and $V \setminus V_s$, implying that $F'$ is an induced edge cut. Assume to the contrary that there is an edge $e \in F'$ with both endpoints in one of the sides, say $V_s$, then $F' \setminus \{e\}$ is still a cut that disconnects $s$ and $t$ (as $V_s$ is still disconnected from the rest of the graph if we add $e$), contradicting the minimality of $F'$. A symmetric argument shows that $e$ cannot have both its endpoints in $V \setminus V_s$.
\end{proof}

We next show that given an induced edge cut $F'$, there is a simple way to determine the two sides of the cut induced by $F'$ (see Figure \ref{cutSidesPic} for illustration). For a vertex $v$ and an induced edge cut $F'$, we denote by $n_v(F')$ the number of edges from $F'$ in the path from the root $r$ to $v$ in the spanning tree $T$. We show the following.

\remove{
\begin{claim}
Let $F'$ be an induced edge cut, and let $T$ be a spanning tree with root $r$. Let $V_0$ be all the vertices $v$ where in the path from $r$ to $v$ there is an even number of edges from $F'$, and let $V_1 = V \setminus V_0$. Then $(V_0,V_1)$ is the induced edge cut defined by $F'$.
\end{claim}
}

\begin{claim} \label{obs_cut_sides}
Let $F'$ be an induced edge cut. Let $$V_0=\{v \in V|\ n_v(F')\ is \ even \},$$ $$V_1=\{v \in V|\ n_v(F')\ is \ odd \}.$$ Then $(V_0,V_1)$ is the induced edge cut defined by $F'$.
\end{claim}

\begin{proof}
Since $F'$ is an induced edge cut, the endpoints of every edge in $F'$ are on different sides of the cut. Hence, if we scan the tree $T$ from the root to the leaves, every time we reach an edge from $F'$ we change the side of the cut. It follows that one side of the cut contains all vertices $v$ such that $n_v(F')$ is even, and the other side has all vertices $v$ such that $n_v(F')$ is odd. Hence $V_0,V_1$ are the two sides of the cut.
\end{proof}

\setlength{\intextsep}{0pt}
\begin{figure}[h]
\centering
\setlength{\abovecaptionskip}{-2pt}
\setlength{\belowcaptionskip}{6pt}
\includegraphics[scale=0.55]{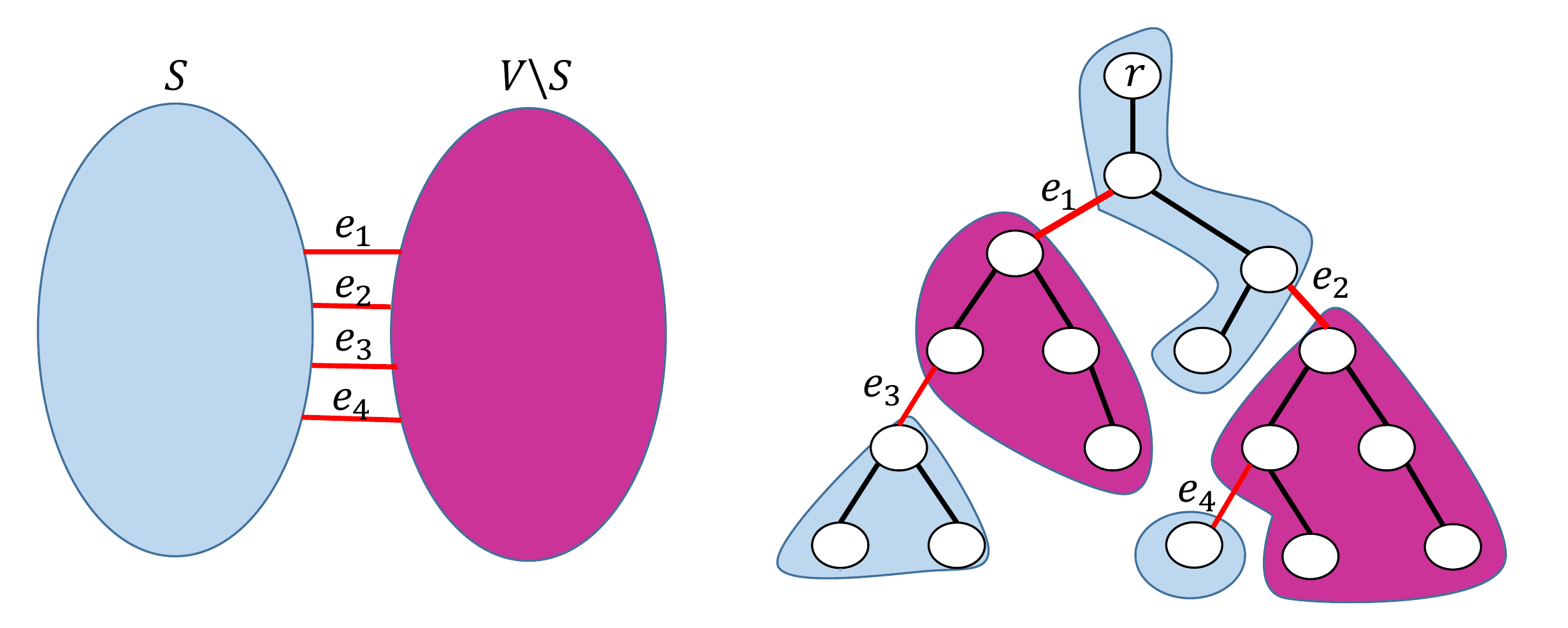}
 \caption{Here $F'=\{e_1,e_2,e_3,e_4\}$ is an induced edge cut. On the right, you can see the partition into sides in the tree. Every time we reach an edge from $F'$, we change the side of the cut.}
\label{cutSidesPic}
\end{figure}

From Claims \ref{obs_induced} and \ref{obs_cut_sides}, we get the following.

\begin{corollary} \label{cor_ft_connectivity}
The vertices $s$ and $t$ are disconnected by $F$ if an only if there is an induced edge cut $F' \subseteq F$, such that one of the values $n_s(F'),n_t(F')$ is even and the other is odd. 
\end{corollary}

This gives a simple approach to detect if $s$ and $t$ are disconnected by $F$. We go over all subsets $F' \subseteq F$, for each one of them we first check if $F'$ is an induced edge cut using the cycle space labels. Second, if $F'$ is an induced edge cut, we compute the values $n_s(F'),n_t(F')$, if the number is even for one of them and odd for the second, we deduce that $F'$ disconnects $s$ and $t$. Note that we can use the ancestry labels to compute the values $n_s(F'),n_t(F')$. For example, for computing $n_s(F')$ we should check how many edges in $F'$ are in the tree path between $r$ to $s$. For this, for each tree edge $e=(u,v)$ in $F'$, we check if it is above $s$ in the tree, which happens if and only if both $u$ and $v$ are ancestors of $s$.
This simple approach requires time exponential in $|F|$ for going over all subsets of $F$, we next show a faster way to check the same condition.

\subsubsection{Faster Decoding Algorithm}

We next show that checking the condition from Corollary \ref{cor_ft_connectivity} boils down to solving a system of linear equations.
First, note that from Lemma \ref{cycle_space_lemma}, w.h.p, a set of edges $F' \subseteq F$ is an induced edge cut iff $\Moplus_{e \in F'} \phi(e) = 0$. Hence, if we want to check if there is a non-empty subset $F' \subseteq F$ that is an induced edge cut it is equivalent to checking if there exists a binary vector $x=(x_1,...,x_f) \neq 0$ such that $\Moplus_{1 \leq i \leq f} x_i \phi(e_i) = 0$, where $\{e_1,...e_f\}$ are the edges of $F$. Or equivalently checking if the vectors $\{\phi(e)\}_{e \in F}$ are linearly dependant. To check the condition from Corollary \ref{cor_ft_connectivity}, we generalize this idea. 

Let $b = O(f + \log{n})$ be the size of the cycle space labels.
Given a triplet $(s,t,F)$, we assign for each edge $e \in F$, a binary vector $\phi'(e)$ of length $b+2$, as follows.
\begin{enumerate}
\item If $e$ is a tree edge which is in the tree path $r-s$ but not in the path $r-t$, then $\phi'(e)=10\phi(e)$.
\item If $e$ is a tree edge which is in the tree path $r-t$ but not in the path $r-s$, then $\phi'(e)=01\phi(e)$.
\item In all other cases, $\phi'(e)=00\phi(e).$ 
\end{enumerate}

We denote by $w_1,w_2$ binary vectors of length $b+2$ such that $w_1=100..0,w_2=010...0$ (all right entries are equal to 0).
We show that the condition from Corollary \ref{cor_ft_connectivity} holds iff there is a binary vector $x=(x_1,...,x_f)$ and $j \in \{1,2\}$ such that $$\Moplus_{1 \leq i \leq f} x_i \phi'(e_i) = w_j.$$ This holds iff there is a solution to at least one of the linear systems $Ax=w_1,Ax=w_2$, where $A$ is a $(b+2) \times f$ matrix that has the vectors $\{\phi'(e)\}_{e \in F}$ as its column vectors, and $x,w_1,w_2$ are column vectors. All operations are modulo 2.

\begin{lemma}
With high probability, the vertices $s$ and $t$ are disconnected by $F$ if an only if there is a binary vector $x=(x_1,...,x_f)$ and $j \in \{1,2\}$ such that $\Moplus_{1 \leq i \leq f} x_i \phi'(e_i)  = w_j.$
\end{lemma}

\begin{proof}
We assume for the proof that the cycle space labels are correct, i.e., a set of edges $F' \subseteq F$ is an induced edge cut iff $\Moplus_{e \in F'} \phi(e) = 0$. This happens w.h.p from Lemma \ref{cycle_space_lemma} and the choice of $b=O(f + \log{n})$.

First we show that if $s$ and $t$ are disconnected by $F$, the condition of the lemma holds.
From Corollary \ref{cor_ft_connectivity}, $s$ and $t$ are disconnected by $F$ iff there is an induced edge cut $F' \subseteq F$, such that one of the values $n_s(F'),n_t(F')$ is even and the other is odd. Denote by $n'_s(F')$ the number of edges from $F'$ in the $r-s$ tree path that are not in the $r-t$ path, and denote by $n'_t(F')$ the number of edges from $F'$ in the $r-t$ tree path that are not in the $r-s$ path. Note that if one of the values $n_s(F'),n_t(F')$ is even and the other is odd, then also one of $n'_s(F'),n'_t(F')$ is even and the other is odd, as if we denote by $y$ the number of edges from $F'$ that are in both $r-s$ and $r-t$, we get that $n'_s(F') = n_s(F') - y, n'_t(F') = n_t(F') - y$. Assume first that $n'_s(F')$ is even and $n'_t(F')$ is odd. Let $x$ be the characteristic vector of $F'$. We show that $\Moplus_{1 \leq i \leq f} x_i \phi'(e_i)  = w_2$. First, as $F'$ is an induced edge cut, we have that $\Moplus_{e \in F'} \phi(e) = 0$. Hence, the $b$ last bits of $\Moplus_{1 \leq i \leq f} x_i \phi'(e_i)$ are equal to 0 as needed. $F'$ has even number of edges that are in the path $r-s$ and not $r-t$, as the labels $\phi'(e)$ of all these edges start in $10$, the XOR of the first 2 bits of these edges sums to $00$.  $F'$ has odd number of edges that are in the path $r-t$ but not $r-s$. The labels of all these edges start in $01$, as there is an odd number of them, the XOR of the first 2 bits of these edges sums to $01$. All other edges have labels that start in $00$, hence the XOR of their first 2 bits sums to $00$. Overall we get that $\Moplus_{1 \leq i \leq f} x_i \phi'(e_i)=\Moplus_{e \in F'} \phi'(e)=010...0=w_2$. The case that $n'_s(F')$ is odd and $n'_t(F')$ is even is symmetric and results in the equation $\Moplus_{1 \leq i \leq f} x_i \phi'(e_i)=100...0=w_1.$

On the other hand, if we have that $\Moplus_{1 \leq i \leq f} x_i \phi'(e_i)=w_j$ for a binary vector $x=(x_1,...,x_f)$ and $j \in \{1,2\}$, we can build from it $F'$ that satisfies the condition in Corollary \ref{cor_ft_connectivity}, as follows. We define $F'$ to be all edges $e_i \in F$ such that $x_i =1$. Since $\Moplus_{1 \leq i \leq f} x_i \phi'(e_i)=\Moplus_{e \in F'} \phi'(e)=w_j$, we have that $\Moplus_{e \in F'} \phi(e) = 0$, hence $F'$ is an induced edge cut. Additionally if $w_j=w_2$, it implies that the XOR of the first 2 bits of labels $\{\phi'(e)\}_{e \in F'}$ are equal to $01$. By the definition of the labels, this can only happen if $n'_s(F')$ is even and $n'_t(F')$ is odd. Similarly, if $w_j=w_1$, then $n'_s(F')$ is odd and $n'_t(F')$ is even. In both cases we get that one of the values $n_s(F'),n_t(F')$ is even and the other is odd, hence $s$ and $t$ are disconnected by $F$ from Corollary \ref{cor_ft_connectivity}. 
\end{proof}

To conclude, the question if $s$ and $t$ are disconnected by $F$ boils down to checking if there is a solution to at least one of the linear systems $Ax=w_1,Ax=w_2$, where $A$ is a $(b+2) \times f$ matrix, and $b=O(f + \log{n})$. Note that we can construct the labels $\phi'(e)$ and hence the matrix $A$ given the labels of $s,t,F$. For this, we need the labels $\phi(e)$ of edges in $F$, and also to distinguish for each edge in $F$ if it is in the $r-s,r-t$ paths in the tree. The latter can be deduced from the ancestry labels of $s,t,F$ and from the bits indicating which edges in $F$ are tree edges. A tree edge $e=(u,v) \in F$ is in the $r-s$ path iff both $u$ and $v$ are ancestors of $s$, this can be checked in $O(1)$ time using the ancestry labels of $u,v,s$. Hence we can build the matrix $A$ in $O(fb)$ time. To check if the linear systems have a solution we can use Gaussian elimination, that takes $O(MN^2)$ time for $M \times N$ matrix, in our case this is $O((f+\log{n})f^2)$. Alternatively, we can use $O(N^{\omega})$ algorithms for $N \times N$ matrices, where $\omega$ is the exponent of matrix multiplication.
For this, we add zero columns to our matrix $A$ to make it a $(b+2) \times (b+2)$ matrix $A'$ and increase the length of $x$ to $b+2$, the new system $A'x=w_i$ has a solution iff the original system $Ax=w_i$ has a solution. The complexity here is $O((b+2)^{\omega})=O((f+\log{n})^{\omega})$.
This gives the following. 
 
\begin{theorem}
There is a randomized $f$-FT connectivity labeling scheme that assigns the edges and vertices of the graph labels of size $O(\log{n})$ bits per vertex and $O(f + \log{n})$ bits per edge. The decoding time of the scheme is $\min\{O((f+\log{n})f^2),O((f+\log{n})^{\omega})\}$. The time complexity for assigning the labels is $\widetilde{O}((m+n)f).$
\end{theorem}  

\remove{
\begin{theorem}
We can assign the edges and vertices of the graph labels of size $O(\log{n})$ bits per vertex and $O(f + \log{n})$ bits per edge, such that given the labels of $(s,t,F)$ we can check if $s$ and $t$ are disconnected by $F$ in $O(f^3 \log{n})$ \mtodo{check the complexity} time, w.h.p.
\end{theorem} 
} 
\subsection{Connectivity Labels Based on Graph Sketches}\label{sec:ftconn-sketch}
In this section, we show the following:
\begin{theorem}
For every undirected $n$-vertex graph $G=(V,E)$, a positive integer $f$, there is a randomized $f$-FT connectivity labels $\FTConnLabel_{G}: V \cup E \to \{0,1\}^{\ell}$ of length $\ell=O(\log^3 n)$ bits. The decoding time of the scheme is $\widetilde{O}(f)$, and the computation time for assigning the labels is $\widetilde{O}(m+n)$.
\end{theorem}
In Section \ref{sec:label-alg}, we present the labeling algorithm which assigns labels based on the notion of graph sketches. In Section \ref{sec:dec-alg} we present the decoding algorithm that given the label information determines if $s$ and $t$ are connected in $G \setminus F$. When the graph $G$ is clear from the context, we may omit it and simply write $\FTConnLabel$.

\subsubsection{The Labeling Algorithm}\label{sec:label-alg}
Given a connected graph $G$, let $T$ be an arbitrary rooted spanning tree in $G$ that is used throughout this section. In our future applications of this labeling scheme (e.g., routing), both the graph $G$ and the tree $T \subseteq G$ will be given as input to the labeling algorithm. In the latter case, we denote the output labels by $\FTConnLabel_{G,T}$. Throughout, all vertices have unique ids $\ID(v)$ between $\{1,\ldots,n \}$.  

\paragraph{Extended Edge Identifiers.} In our algorithm it is important to distinguish between an identifier of a single edge to the bitwise XOR of several edges. For this purpose, we define for each edge $e$ an extended edge identifier $\EID_T(e)$ that allows distinguishing between these cases, and serves as the identifier of the edge.
The extended edge identifier $\EID_T(e)$ consists of a (randomized) unique distinguishing identifier $\UID(e)$, as well as additional tree related information that facilitates the decoding procedure. The computation of $\UID(e)$ is based on the notion of $\epsilon$-\emph{bias} sets \cite{naor1993small}. The construction is randomized and guarantees that, w.h.p., the XOR of the $\UID$ part of each given subset of edges $S \subseteq E$, for $|S|\geq 2$, is not a legal $\UID$ identifier of any edge.
Let $\XOR(S)$ be the bitwise XOR of the extended identifiers of edges in $S$, i.e., $\XOR(S)=\oplus_{e \in S} \EID_T(e)$. In addition, let $\XOR_U(S)=\oplus_{e \in S} \UID(e)$. Missing proofs are deferred to Appendix \ref{sec:miss-proof}.

\begin{lemma}[Modification of Lemma 2.4 in \cite{GhaffariP16}]
\label{cl:epsbias}
There is an algorithm that creates a collection $\mathcal{I}=\{\UID(e_1), \ldots, \UID(e_{M})\}$ of $M=\binom{n}{2}$ random identifiers for all possible edges $(u,v)$, each of $O(\log n)$-bits using a seed $\mathcal{S}_{ID}$ of $O(\log^2 n)$ bits. These identifiers are such that for each subset $E' \subseteq E$, where $|E'|\neq 1$, we have $\Pr[\XOR_U(E') \in \mathcal{I}] \leq 1/n^{10}$. In addition, given the identifiers $\ID(u), \ID(v)$ of the edge $e=(u,v)$ endpoints, and the seed $\mathcal{S}_{ID}$, one can determine $\UID(e)$ in $\widetilde{O}(1)$ time.
\end{lemma}
\def\APPENDUNIQUEID{
\begin{proof}[Proof of Lemma \ref{cl:epsbias}]
The lemma is proved in \cite{GhaffariP16}, the only part that is not discussed there is the time to determine $\UID(e)$ that follows from \cite{naor1993small}. 
By Theorem 3.1 of \cite{naor1993small}, given the seed $\mathcal{S}_{ID}$ and the edge identifier $e_j=(\ID(u), \ID(v))$, determining the $i^{th}$ bit of $\UID(e_{j})$ can be done in $O(\log n)$ time. Thus, determining all $O(\log n)$ bits, takes $O(\log^2 n)$ time. 
\end{proof}
}
For every vertex $v \in G$, let $\LCALabel_T(v)$ be the ancestor label of $v$ computed for the given tree $T$ using Lemma \ref{anc_labels}. The extended identifier $\EID_T(e)$ is given by
\begin{equation}\label{eq:extend-ID}
\EID_T(e)=[\UID(e), \ID(u), \ID(v), \LCALabel_T(u), \LCALabel_T(v)]~.
\end{equation}
The identifiers of $\ID(u), \ID(v)$ are used in order to verify the validity of the unique identifier $\UID(e)$.  
When the tree $T$ is clear from the context, we might omit it and simply write $\EID(e)$. As we will see, the labeling scheme will store the seed $\mathcal{S}_{ID}$ as part of the labels of the tree edges. 


\paragraph{Fault-Tolerant Labels via Graph Sketches.} 
Graph sketches are a tool to identify outgoing edges. We start by providing an intuition for them. Say that $S$ is a connected component, and that there are $2^j$ edges outgoing from $S$. If we sample all edges in the graph with probability $1/2^j$, there is a constant probability that exactly one outgoing edge from $S$ is sampled, and our goal is to find it using local information stored at the vertices of $S$. This information is the \emph{sketch}. 
The sketch of each vertex stores the bitwise XOR of sampled edges adjacent to it. Now looking at the XOR of all the sketches of vertices of $S$ allows to detect an outgoing edge. This holds as any sampled edge that has both endpoints in $S$ gets cancelled out, and we are left with the XOR of sampled edges outgoing from $S$. If there is exactly one outgoing edge, we find it. To increase the success probability we can repeat the process $O(\log{n})$ times. We define sets of vertices $E_{i,j}$, where for $i \in \{1, \ldots, c \log n\}$, the set $E_{i,j}$ is obtained by sampling each edge with probability $2^{-j}$. Since we repeat the process $O(\log{n})$ times for each $j$, then w.h.p we can use the sketches to identify outgoing edge from any component. To use this approach in our context, it is crucial to be able to simulate the sampling process using a small random seed. To do this, we follow \cite{DuanConnectivityArxiv16,DuanConnectivitySODA17} and use pairwise independent hash functions to decide whether to include edges in sampled sets.
We choose $L=c\log n$ 
pairwise independent hash functions $h_1, \ldots, h_{L}:\{0,1\}^{\Theta(\log n)} \to \{0, \ldots, 2^{\log m}-1\}$, 
and for each $i \in \{1, \ldots, L\}$ and $j \in [0,\log m]$, define the edge set 
$$E_{i,j} =\{ e \in E ~\mid~ h_i(e) \in [0,2^{\log m-j})\}~.$$ 
Each of these hash functions can be defined using a random seed of logarithmic length \cite{TCS-010}. Thus, a 
random seed $\mathcal{S}_h$ of length $O(L \log n)$ can be used to determine the collection of all these $L$ functions. As observed in \cite{DuanConnectivityArxiv16,GibbKKT15}, pairwise independence is sufficient to guarantee that for any set $E' \subset E$ and any $i$, there exists an index $j$, such that with constant probability $\XOR(E' \cap E_{i,j})$ is the name (extended identifier) of one edge in $E'$, for a proof see Lemma 5.2 in  \cite{GibbKKT15}.
\begin{lemma}\label{lem:hitting-pairwise}
For any edge set $E'$ and any $i$, with constant probability there exists a $j$ satisfying that $|E' \cap E_{i,j}|=1$.
\end{lemma}

We also need to be able to tell that a bit string of $\XOR(E' \cap E_{i,j})$ is a legal edge ID or not. Here we exploit the extended ids. See Appendix \ref{sec:miss-proof} for a proof.
\begin{lemma} \label{lemma_unique}
Given the seed $\mathcal{S}_{ID}$, one can determine in $\widetilde{O}(1)$ time if $\XOR(E' \cap E_{i,j})$ corresponds to a single edge ID in $G$ or not, w.h.p.
\end{lemma}
\def\APPENDLEMMUNIQUE{
\begin{proof}[Proof of Lemma \ref{lemma_unique}]
Let $X=\XOR(E' \cap E_{i,j})$. Letting $E''=E' \cap E_{i,j}$, then $X$ can be written as the concatenation of $\XOR_1(E'')$ and $\XOR_2(E'')$, where $\XOR_1(E'')=\XOR_U(E'')$ is the bit-wise XOR of the unique identifiers $\UID(e)$ for $e \in E''$ and $\XOR_2(E'')$ is the bit-wise XOR of the remaining information in the extended identifiers of $E''$.  We now show how using the seed and $\XOR_2(E'')$, one can test the validity of $\XOR_1(E'')$.
The algorithm detects the case that $|E''| \geq 2$ as follows. First, in the case that $E''$ is a single edge, $\XOR_2(E'')$ should contain legal ids $\ID(u),\ID(v)$. If this is not the case, it follows that $|E''| \neq 1$. If $\XOR_2(E'')$ contains legal ids $\ID(u),\ID(v)$, we use them and the seed $\mathcal{S}_{ID}$ to determine $\UID(e)$ for $e = (u,v)$, and we check if $\XOR_1(E'')=\ID_1(e)$. We have two options, either $E'' = \{e\}$ is the single edge $e$, in which case $\XOR_U(E'')=\UID(e) \in \mathcal{I}$, and the verification succeeds. Otherwise $|E''| \geq 2$, in which case, from Lemma \ref{cl:epsbias}, $\Pr[\XOR_U(E'') \in \mathcal{I}] \leq 1/n^{10}$, hence w.h.p $\XOR_U(E'') \neq \UID(e) \in \mathcal{I}$ and we identify that $|E''| \geq 2$.
\end{proof}
}

For each vertex $v$ and indices $i,j$, let $E_{i,j}(v)$ be the edges incident to $v$ in $E_{i,j}$. 
The $i^{th}$ \emph{basic sketch unit} of each vertex $v$ is then given by:
\begin{equation}
\label{eq:vsketch}
\Sketch_{G,i}(v)=[\XOR(E_{i,0}(v)),\ldots,\XOR(E_{i,\log m}(v))].
\end{equation}
The sketch of each vertex $v$ is defined by a concatenation of $L=\Theta(\log n)$ basic sketch units: 
$$\Sketch_G(v)=[\Sketch_{G,1}(v),\Sketch_{G,2}(v), \ldots\Sketch_{G,L}(v)]~.$$ 
For every subset of vertices $S$, let 
$\Sketch_G(S)=\oplus_{v \in S}\Sketch_G(v).$ When the graph $G$ is clear from the context, we may omit it and write $\Sketch_{i}(v)$ and $\Sketch(v)$. 

We are now ready to define the fault-tolerant connectivity labels of vertices and edges. 
The label of each vertex $u$ is given by:
\begin{equation}\label{eq:conn-vertex}
\FTConnLabel_{G,T}(u)=\langle \LCALabel_T(u), \ID(u) \rangle~,
\end{equation}
where $\LCALabel_T(u)$ is the ancestry label of $u$ with respect to the tree $T$. 
For every $u \in V(T)$, let $T_u$ be the subtree rooted at $u$.  The label $\FTConnLabel_{G,T}(e)$ of each \emph{edge} $e=(u,v)$ is given by:
\begin{equation*}
    \FTConnLabel_{G,T}(e)=
    \begin{cases}
      \langle \EID_T(e), \Sketch(V(T_u)), \Sketch(V(T_v)), \Sketch(V), \mathcal{S}_{ID}, \mathcal{S}_h\rangle ,& \mbox{~for~} e \in T \\
     \langle \EID_T(e) \rangle,& \mbox{~Otherwise}.
    \end{cases}
\end{equation*}

We complete this subsection by bounding the label size and computation time of the labeling algorithm. For proofs see Appendix \ref{sec:miss-proof}. 
\begin{claim}\label{cl:label-length}
The label length is $O(\log^3 n)$ bits.
\end{claim}
\def\APPENDLABELCONSISE{
\begin{proof}[Proof of Claim \ref{cl:label-length}]
The label size is dominated by the sketching information $\Sketch(V(T_u))$, which is made of a concatenation of the bitwise XOR of $O(\log n)$ basic sketch units $\Sketch_i(u)$. By Eq. (\ref{eq:vsketch}), each unit has $O(\log^2 n)$ bits, and thus overall, the label has $O(\log^3 n)$ bits.
\end{proof}
}

We show that assigning the labels takes  $\widetilde{O}(m+n)$ time.
\begin{claim}\label{cl:time-conn-labelsketch}
The time complexity of the labeling algorithm is $\widetilde{O}(m+n).$
\end{claim}
\def\APPENDCONNLABELSKETCH{
\begin{proof}[Proof of Claim \ref{cl:time-conn-labelsketch}]
To compute the labels of vertices we assign ids to vertices in $O(n)$ time, and compute ancestry labels in $O(n)$ time using Lemma \ref{anc_labels}. To compute the extended identifiers $\EID_T(e)$, we also choose the random seed $\mathcal{S}_{ID}$ and compute $\UID(e)$ using Lemma \ref{cl:epsbias}, this takes $\widetilde{O}(1)$ time per edge, and $\widetilde{O}(m)$ time for all edges. Lastly, we should compute the sketch values $\Sketch(V(T_u))$. For this, first, we choose the random seed $\mathcal{S}_h$, and compute the values $\Sketch_G(v)$. For this, we should identify for each vertex the adjacent edges in $E_{i,j}$. For each edge we can identify the sets it belongs to in $\widetilde{O}(1)$ time using Fact \ref{fc:pairwise}. This allows us computing the sketch values of all vertices in $\widetilde{O}(m+n)$ time. We can then compute the values $\Sketch(V(T_u))$ by scanning the tree in $\widetilde{O}(n)$ time.   
\end{proof}
}
Finally, the subsequent decoding algorithm will be based on the following useful property of the graph sketches, stored by our labels. 
\begin{lemma}\label{lem:sketch-property}
For any subset $S$, given one basic sketch unit $\Sketch_i(S)$ and the seed $\mathcal{S}_{ID}$ one can compute, with constant probability, an outgoing edge $E(S, V \setminus S)$ if such exists. The complexity is $\widetilde{O}(1)$ time.
\end{lemma}
\def\APPENDSKETCHPROP{
\begin{proof}[Proof of Lemma \ref{lem:sketch-property}]
The proof follows from Lemma \ref{lem:hitting-pairwise}. Note that by definition of the sketch values $\Sketch_i(S)=\oplus_{v \in S}\Sketch_i(v)=[\XOR(E_{i,0}(S)),\ldots,\XOR(E_{i,\log m}(S))],$ where $E_{i,j}(S)$ are the outgoing edges from $S$ in $E_{i,j}$ (edges that have both endpoints in $S$ are cancelled out by the XOR operation). Let $E'$ be all the outgoing edges from $S$. From Lemma \ref{lem:hitting-pairwise}, with constant probability there exists a $j$ such that $|E' \cap E_{i,j}|=1$. In this case, $\XOR(E_{i,j}(S))$ corresponds to an extended id of a single outgoing edge from $S$. We can check if this happens in $\widetilde{O}(1)$ time using Lemma \ref{lemma_unique}.   
\end{proof}
}

\subsubsection{The Decoding Algorithm} \label{sec:dec-alg}
We next describe the decoding algorithm where given a triplet $s,t, F \in V \times V \times E^f$ along with their labels, it determines whether $s$ and $t$ are connected in $G\setminus F$, w.h.p. 
The decoding algorithm has four key steps: The first step identifies the at most $f+1$ components $\mathcal{C}_0=\{C_1,\ldots, C_\ell\}$ of $T \setminus F$, as well as the components of $s$ and $t$ in $\mathcal{C}_0$. The second step uses the label information to compute the sketch value $\Sketch(C_i)$ of each component $C_i \in \mathcal{C}_0$. The third step modifies this sketch information into $\Sketch_{G \setminus F}(C_i)$, by subtracting the information related to the faulty edges. The forth and final step uses the sketch information in order to simulate $L=O(\log n)$ steps of the Boruvka algorithm. At the end of these steps, the decoding algorithm identifies the connected components of both $s$ and $t$ in $G \setminus F$. In the case where $s$ and $t$ are indeed connected in $G \setminus F$, the algorithm also outputs a succinct representation of an $s$-$t$ path in $G \setminus F$. This extra information would be used later on by our compact routing scheme. We next describe these steps in details. 

\paragraph{Step 1: Identification of the connected components $\mathcal{C}_0$ in $T \setminus F$.} 
Let $F_T=F \cap T$ be the faulty tree edges and let $F_{NT}=F \setminus F_T$ be the faulty non-tree edges. Let $Q=\{s,t\} \cup V(F_T)$.  Each component $C_i$ of $T \setminus F$ will be identified by the maximum vertex ID in $C_i \cap V(F_T)$. Note that in the case where $F_T=\emptyset$, $T \setminus F=T$ and thus $s$ and $t$ are connected iff $s,t \in V(T)$.  From now on, we therefore assume that $F_T \neq \emptyset$. 

We next show that although we do not have full information about the tree $T$ and the vertices of each connected component, the ancestry labels of $V(F_T)$ give us enough information to identify the connected components of $T \setminus F$. Additionally, given an ancestry label of a vertex $u$, we can identify the connected component of $u$. To obtain this, it is helpful to look at the \emph{component tree} that is obtained by contracting each connected component of $T \setminus F$ to one vertex, as follows. Let $\ell = |F_T|+1.$ The component tree $T_C = (\mathcal{C}_0, E_C)$ is a tree of $\ell$ vertices representing the connected components in $T \setminus F$, and $|F_T|=\ell-1$ edges corresponding to the edges of $F_T$. There is an edge $(C_i,C_j) \in E_C$ iff there is an edge $(u,v) \in F_T$ where $u \in C_i, v \in C_j$. See Figure \ref{componentTreePic} for an illustration.

\setlength{\intextsep}{0pt}
\begin{figure}[h]
\centering
\setlength{\abovecaptionskip}{-2pt}
\setlength{\belowcaptionskip}{6pt}
\includegraphics[scale=0.55]{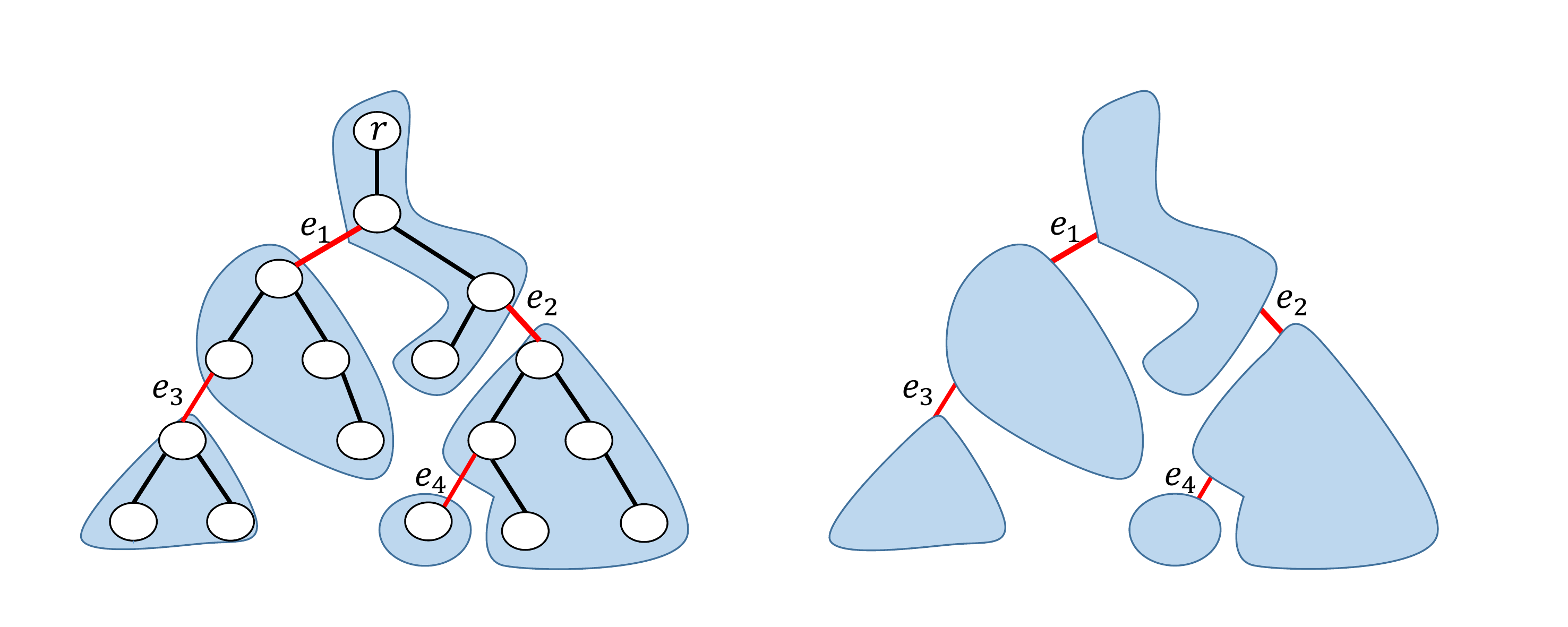}
 \caption{Illustration of the component tree where $F=\{e_1,e_2,e_3,e_4\}$. Each connected component of $T \setminus F$ is contracted to one vertex on the right.}
\label{componentTreePic}
\end{figure}

We can construct the tree $T_C$ using the ancestry labels of the edges $F_T$.  For this, for each edge $e \in F_T$ we just need to identify the set of edges from $F_T$ above $e$ in $T$. Moreover, for a given vertex $v$, its connected component is exactly determined by the set of edges in $F_T$ above it in $T$, which can again be identified using the ancestry labels of $v \cup V(F_T)$. In particular, we can identify the connected components of $s$ and $t$. The component tree can be constructed in $O(f^2)$ time by checking for any pair of edges $e,e' \in F_T$, if $e$ is above $e'$ in the tree. We next show a faster algorithm taking only $\widetilde{O}(f)$ time by exploiting properties of the ancestry labels. Moreover, we show that the component of each vertex can be identified in $O(\log{f})$ time.

\begin{claim} \label{claim_component_tree}
The component tree can be constructed in $O(f \log{f})$ time. Additionally, given $\LCALabel_T(v)$, we can identify the connected component of $v$ in $T \setminus F$ in $O(\log{f})$ time.
\end{claim}

\begin{proof}
Our algorithm uses ancestry labels based on DFS from \cite{kannan1992implicat}. In this scheme, the label of each vertex $v$ is composed of two numbers $(DFS_1(v),DFS_2(v))$ that represent the first and last times a DFS scan of the tree visits $v$. A vertex $u$ is an ancestor of a vertex $v$ iff the interval $(DFS_1(u),DFS_2(u))$ contains the interval $(DFS_1(v),DFS_2(v))$. To build the component tree, we sort the labels of $V(F_T)$, as described next. First, for each component $C \in T \setminus F$, we use the highest vertex in the component to represent the component. For the highest component, this is the root $r$. For any other component, we have that the highest vertex of the component, $v$, is in $V(F_T)$. This holds as the edge connecting $v$ to its parent $p(v)$ is necessarily in $F_T$ (otherwise, $v$ is not the highest vertex in its component), see Figure \ref{componentTreePic} for illustration. Hence, for any edge $(v,p(v)) \in F_T$, we have that the vertex $v$ represents one component (we can identify which of the vertices is the parent using the ancestry labels). Hence, we have $|F_T|+1$ vertices $v_i$ representing the components $C_i$ of the component tree, and we also know the ancestry labels $(DFS_1(v_i),DFS_2(v_i))$ of all vertices $v_i$, except $r$. For $r$ we can use the label $(1,M)$ where $M$ is a number greater than all values $DFS_2(v_i)$ of other vertices. We next use these labels to determine the structure of the component tree.
For this, we create for each vertex $v_i$ two tuples: $(DFS_1(v_i),v_i,1),(DFS_2(v_i),v_i,2)$, and we sort the $2(|F_T|+1)$ tuples according to their first coordinate. This takes $O(f \log{f})$ time. We next scan the sorted list, and when we reach the tuple $(DFS_1(v_i),v_i,1)$, we identify the parent of $v_i$ in the component tree, as follows. The first tuple is $(1,r,1)$ and $r$ is set to be the root of the component tree. For a vertex $v_i \neq r$, we identify its parent when we reach $(DFS_1(v_i),v_i,1)$. Let $(DFS_b(u),u,b)$ be the last tuple before $(DFS_1(v_i),v_i,1)$ in the sorted order. If $b=1$, then $u$ is the parent of $v_i$ in the component tree. If $b=2$, let $w$ be the parent of $u$ in the component tree, then $w$ is also the parent of $v$ in the component tree. Additionally, $w$ was already computed as $(DFS_1(u),u,1)$ appears before $(DFS_1(v_i),v_i,1)$. Hence, we can find the parent of $v$ in $O(1)$ time using the tuple before it. Scanning the list takes $O(f)$ time, and after it we know for each component its parent in the component tree, which gives the complete structure of the tree. We next prove the correctness of the algorithm.  

We first discuss the case that $b=1$. Here $(DFS_1(u),u,1)$ is the last tuple before $(DFS_1(v_i),v_i,1)$. This means that $u$ is necessarily an ancestor of $v$, because the entry $(DFS_1(v_i),v_i,1)$ is between the entries $(DFS_1(u),u,1)$ and $(DFS_2(u),u,2)$, and the DFS scan traverses exactly the subtree of $u$ in the time interval $(DFS_1(u),DFS_2(u))$, implying that $v_i$ is a child of $u$. Moreover, this is the closest ancestor to $v_i$ among the vertices $\{v_1,v_2,...,v_{\ell}\} \setminus \{v_i\}$, as the DFS scan traverses the ancestors of $v_i$ from the highest to the lowest. It follows that $u$ represents the closest component $C$ above $v_i$ in the component tree, as needed.  

We next discuss the case that $b=2$. Here $(DFS_2(u),u,2)$ is the last tuple before $(DFS_1(v_i),v_i,1)$. Note that now $u$ is not an ancestor of $v_i$, as the DFS scan finished scanning the subtree of $u$ before reaching $v_i$, but we claim that $u$ and $v_i$ have the same parent in the component tree. For this, we show they have exactly the same ancestors in the set $\{v_1,v_2,...,v_{\ell}\} \setminus \{u,v_i\}.$ For any ancestor $w\neq u$ of $u$, we have that $DFS_1(w) < DFS_1(u) < DFS_2(u) < DFS_2(w)$. As $(DFS_1(v_i),v_i,1)$ is the first tuple after $(DFS_2(u),u,2)$, it must hold that $DFS_1(w) < DFS_1(v_i) <  DFS_2(w)$, implying that $v_i$ is a child of $w$ as needed. Similarly, any ancestor $w \neq v_i$ of $v_i$ is also an ancestor of $u$, as we have $DFS_1(w) < DFS_2(u) < DFS_1(v_i) < DFS_2(v_i) < DFS_2(w)$.
Hence, the parent of $u$ in the component tree is also the parent of $v_i$ in the component tree, as needed.  
 
Lastly, we show that using similar ideas we can also identify the component of a vertex $v$ in $T \setminus F$. We create for $v$ the tuple, $(DFS_1(v),v,1)$, and use binary search to find the last tuple smaller or equal to it in the sorted list we computed before, denote it by $(DFS_b(u),u,b)$. 
If $b=1$ then $v$ is in the component of $u$, and else it is in the component of the parent of $u$ (that was computed before). The complexity of the binary search is $O(\log{f})$, we next prove correctness. 
One special case is that $v$ is a root of one of the components in the component tree. In this case, the entry $(DFS_b(u),u,b)$ we find is equal to $(DFS_1(v),v,1)$, and $u=v$ is indeed the component of $v$. Otherwise, $v$ is an internal vertex in its component, and the root of the component is the closest ancestor to $v$ in $\{v_1,...,v_{\ell}\}$.
If $b=1$, then as shown before, $u$ is the closest ancestor to $v$ in the component tree, as needed. If $b=2$, then as shown before, $u$ is not an ancestor of $v$, but has exactly the same ancestors in the component tree. Hence, the root $w$ of the component above $u$ is the root of the component of $v$, as needed.    
\end{proof}

\paragraph{Step 2: Computing the sketch values of each component $\mathcal{C}_0$ in $G$.} 
For each component $C_j \in \mathcal{C}_0$ the algorithm computes $\Sketch_G(C_j)$ using the sketch information of the vertices in $V(F_T)$.  The basic observation here is the following. Given $S' \subset S$ and $\Sketch(S), \Sketch(S')$, it holds that $\Sketch(S \setminus S')=\Sketch(S) ~\oplus~ \Sketch(S')$. To compute the sketch values, first, we define for each component a temporary value $\Sketch'_G(C_j)$ as follows. Let $v_j$ be the highest vertex (closest to the root in $T$) in the component $C_j$. For the component of the root $r$, this is $r$. For any other component $C_j$, let $(C_j,p(C_j))$ be the edge connecting $C_j$ to its parent in the component tree. This edge corresponds to an edge $(v_j,p(v_j)) \in F_T$, where $v$ is the highest vertex in $C_j$. We define $\Sketch'_G(C_j) = \Sketch_G(V(T_{v_j}))$. 
Since $(v_j,p(v_j)) \in F_T$, the sketch information $\Sketch'_G(C_j)$ can be obtained from the label of the tree edge $(v_j,p(v_j))$. We also know the temporary sketch value of the component of $r$, as $\Sketch_G(V_{r})=\Sketch_G(V)$ is part of the labels of all tree edges (and we assume that $F_T \neq \emptyset$). We next use the temporary sketch values to compute the sketch values of components using the following claim.

\begin{claim}
Let $C_j$ be a component in $T \setminus F$. If $C_j$ is a leaf in the component tree, we have $\Sketch_G(C_j) = \Sketch'_G(C_j).$ Otherwise, let $D=\{D_1,...,D_t\}$ be the children of $C_j$ in the component tree and let $\Sketch'(D)=\oplus_{1 \leq i \leq t} \Sketch'_G(D_i)$, then $\Sketch_G(C_j) = \Sketch'_G(C_j) \oplus \Sketch'(D).$ 
\end{claim}

\begin{proof}
It holds that $\Sketch_G(C_j) = \oplus_{v \in C_j} \Sketch_G(v)$. By definition, $\Sketch'_G(C_j)=\Sketch_G(V(T_{v_j})) = \oplus_{v \in V(T_{v_j})} \Sketch_G(v)$ is the XOR of sketches of all vertices in the subtree of $v_j$. As $v_j$ is the highest vertex in $C_j$, if $C_j$ is a leaf component in the component tree, then the vertices in $C_j$ are exactly the vertices in $T_{v_j}$, and the claim follows. Otherwise, the vertices in $C_j$ are all vertices in $T_{v_j}$ that are not contained in any component below $C_j$. Hence, to compute the value $\Sketch_G(C_j)$, we should subtract from $\Sketch_G(V(T_{v_j}))$ the sketch values of vertices in components below $C_j$. Let $D_1,\ldots,D_t$ be the children of $C_j$ in the component tree, and let $u_1,\ldots,u_t$ be the highest vertices in the components $D_1,\ldots,D_t$, respectively. Any vertex that is in some component below $C_j$ is in exactly one of the subtrees $T_{u_1},\ldots,T_{u_t}$. Hence the sketch value of vertices in components below $C_j$ equals $\oplus_{1 \leq i \leq t} \Sketch_G(V(T_{u_i}))= \oplus_{1 \leq i \leq t} \Sketch'_G(D_i)=\Sketch'(D)$. To conclude, we get $\Sketch_G(C_j)=\Sketch_G(V(T_{v_j})) \oplus \Sketch'(D)=\Sketch'_G(C_j) \oplus \Sketch'_G(D)$, as needed.
\end{proof}

To conclude, from the values $\Sketch'_G(C_j)$, we can easily compute the values $\Sketch_G(C_j)$. The complexity is $\widetilde{O}(f)$, as for each component, the sketch $\Sketch'(C_j)$ participates in two computations, and we have at most $O(f)$ components and the sketches have poly-logarithmic size.

\paragraph{Step 3: Computing the sketch values of each component $\mathcal{C}_0$ in $G \setminus F$.} 
For each faulty edge $e \in F$ (both tree and non-tree edges), our goal is to subtract the sketch information of $e$ from the corresponding components of the endpoint of $e$. The step does not require the label information of the edges, and it would be sufficient to know only the seed $\mathcal{S}_h$ that determines the sampling of edges into the sketches, and the extended identifier of the failing edges. Since $F_T \neq \emptyset$, the algorithm holds the seed $\mathcal{S}_h$ (from the label of an edge $e \in F_T$), and it has the extended identifiers of all edges in $F$ as part of their labels. 

Using the extended identifier of the faulty edge $e=(u,v)$, one can determine in $O(\log{f})$ time the components in $\mathcal{C}_0$ to which its endpoints belong, from Claim \ref{claim_component_tree}. Using the identifier $\EID(e)$ and the seed $\mathcal{S}_h$, one can determine all the indices of the sketch to which the edge $e$ was sampled in $\widetilde{O}(1)$ time using Fact \ref{fc:pairwise}. 
Letting $C_u, C_v$ be the components of $u$ and $v$ in $T \setminus F$, respectively. If $C_u \neq C_v$, then the values $\Sketch_G(C_u),\Sketch_G(C_v)$ are updated by XORing them with the matrix that contains the extended identifier $\EID(e)$ in the relevant positions. The complexity is poly-logarithmic, as the matrix has poly-logarithmic size. In the case that $C_u=C_v$, as $e$ is an internal edge in the component, it is not part of $\Sketch_G(C_u)$, and there is no need to update the value. Overall, doing the computation for all edges in $F$ takes $\widetilde{O}(f)$ time.
From that point on, all sketches of the components $\mathcal{C}_0$ can be treated as sketches that have been computed in $G \setminus F$.

\paragraph{Step 4: Simulating the Boruvka algorithm.} Finally, our goal is to determine the identifiers of the maximal connected components of $s$ and $t$ of $G \setminus F$. The input to this step is the identifiers of the components $\mathcal{C}_0=\{C_{1}, \ldots, C_k\}$ in $T \setminus F$, along with their sketch information in $G \setminus F$. While the algorithm does not have information on the vertices of each component, it knows the component identifier of each vertex in $Q$. 

The algorithm consists of $L=O(\log n)$ phases of the Boruvka algorithm. Each phase $i \in \{1,\ldots, L\}$ will be given as input a partitioning $\mathcal{C}_i=\{C_{i,1}, \ldots, C_{i,k_i}\}$ of (not necessarily maximal) connected components in $G \setminus F$.
These components are identified by an $O(\log n)$ bit identifier, where for each vertex in $Q$, the algorithm receives its unique component identifier in  $\mathcal{C}_i$. In addition, the algorithm receives the sketch information of the components $\mathcal{C}_i$ in $G \setminus F$. The output of the phase is a partitioning $\mathcal{C}_{i+1}$, along with their sketch information in $G \setminus F$ and the identifiers of the components for each vertex in $U$. A component $C_{i,j} \in \mathcal{C}_i$ is \emph{growable} if it has at least one non-faulty outgoing edge to a vertex in $V \setminus C_{i,j}$. That is, the component is growable if it is strictly contained in some maximal connected component in $G \setminus F$. Letting $N_i$ denote the number of growable components in $\mathcal{C}_i$, the output partitioning $\mathcal{C}_{i+1}$ of the $i^{th}$ step guarantees that $N_{i+1}\leq N_i /2$ w.h.p. To obtain outgoings edges from the growable components in $\mathcal{C}_i$, the algorithm uses the $i^{th}$ basic-unit sketch $\Sketch_i(C_{i,j})$ of each $C_{i,j} \in \mathcal{C}_i$. By Lemma \ref{lem:sketch-property}, from every growable component in $\mathcal{C}_i$, we get one outgoing edge $e'=(x,y)$ with constant probability. Using the extended edge identifier of $e'$ the algorithm can also detect the component $C_{i,j'}$ to which the second endpoint, say $y$, of $e'$ belongs using Claim \ref{claim_component_tree}. 
That label allows us to compute the component of $y$ in the initial partitioning $T \setminus F$, i.e., the component $C_{0,q}$ of $y$ in $\mathcal{C}_0$. Thus $y$ belongs to the unique component $C_{i,j'} \in \mathcal{C}_i$ that contains 
$C_{0,q}$.


As noted in prior works \cite{ahn2012analyzing,kapron2013dynamic,DuanConnectivityArxiv16}, it is important to use fresh randomness (i.e., independent sketch information) in each of the Boruvka phases. The reason is that the cut query, namely, asking for a cut edge between $S$ and $V \setminus S$, should not be correlated with the randomness of the sketches. Note that indeed the components of $\mathcal{C}_i$ are correlated with the randomness of the first $(i-1)$ basic sketch units of the vertices. Thus, in phase $i$ the algorithm uses the $i^{th}$ basic sketch units of the vertices (which are independent of the other sketch units) to determine the outgoing edges of the components in $\mathcal{C}_i$.

%

The algorithm then computes the updated sketches of the merged components. This is done by XORing over the sketches of the components in $\mathcal{C}_i$ that got merged into a single component in 
$\mathcal{C}_{i+1}$. In expectation, the number of growable components is reduced by factor $2$ in each phase. Thus after $O(\log n)$ phases, the expected number of growable components is at most $1/n^5$, and using Markov inequality, we conclude that w.h.p there are no growable components. The final partitioning $\mathcal{C}_L$ corresponds w.h.p to the maximal connected components in $G \setminus F$. The pair $s$ and $t$ are connected in $G \setminus F$ only if the components $C_s,C_t$ of $s,t$ respectively in $T \setminus F$ are connected in the final component decomposition.
We next show that the complexity of the algorithm is $\widetilde{O}(f)$. This is also the decoding time of the whole algorithm, as all steps take $\widetilde{O}(f)$ time, as discussed above. 

\begin{claim}\label{cl:complexity-step-four}
The complexity of step 4 is $\widetilde{O}(f)$.
\end{claim}

\begin{proof}
The algorithm has $O(\log{n})$ phases, where in each phase the following is computed. First, given the sketch values of the current components we identify outgoing edges from the components. This takes $\widetilde{O}(1)$ time per component from Lemma \ref{lem:sketch-property}, and $\widetilde{O}(f)$ time for all components, as we have at most $f+1$ components. Next, for each outgoing edge we identify the components it connects using its ancestry labels, this takes $\widetilde{O}(1)$ time per edge using Claim \ref{claim_component_tree}. Then, we merge components accordingly and compute the sketch values of the new components by XORing the sketch values of merged components. Overall this takes $\widetilde{O}(f)$ time, as we have at most $O(f)$ merges. In more detail, we can use a union-find data structure to implement the merges, where every time we merge components we compute the sketch value of the new component. We also maintain for each original component in $T \setminus F$ its current component in phase $i$, this allows us to learn the current components connected by an outgoing edge $e$. This information can be maintained as follows. Let $C$ be a component in $T \setminus F$, and assume we know the component $C_{i,j}$ it belongs to at the beginning of phase $i$. After the merges of phase $i$, $C_{i,j}$ joins some component $C_{i+1,j'}$ of phase $i+1$. We can use the find operation to identify the id of the new component. Overall, we have $O(f)$ merges and $O(f)$ find operations to identify for each component $C \in T \setminus F$, the corresponding component $C_{i+1,j'}$ it belongs to, hence the complexity is bounded by $\widetilde{O}(f)$. 
\end{proof}

%
Finally, we show that the decoding algorithm can be slightly modified to output a compressed encoding of an $s$-$t$ path in $G \setminus F$, using $O(f\log n)$ bits. This encoding is represented by an $s$-$t$ path $\widehat{P}$ that has two type of edges, appearing in an alternate manner on  $\widehat{P}$: $G$-edges and edges $e'=(u,v)$ such that the $u$-$v$ tree path is intact in $T \setminus F$. See Figure \ref{fig:succ-paths}. 
\begin{lemma}\label{lem:useful-recovery-edges}
Consider a triplet $s,t,F$ such that $s$ and $t$ are connected in $G \setminus F$. 
The decoding algorithm can also output a set of at most $f$ recovery edges $R$ such $(T \setminus F) \cup R$ is a spanning tree. In addition, it outputs a labeled $s$-$t$ path $\widehat{P}$ of length $O(f)$ that provides a succinct description of the $s$-$t$ path. The edges of $\widehat{P}$ are labeled by $0$ and $1$, where $0$-labeled edges correspond to $G$-edges and $1$-labeled edges $e=(x,y)$ correspond to $x$-$y$ paths in $T \setminus F$. 
\end{lemma}
\begin{proof}
Let $C_s, C_t$ be the components of $s$ and $t$ in the initial partitioning $\mathcal{C}_0$. In Step $4$ of the decoding algorithm, the Boruvka algorithm is simulated up to the point that $C_s$ and $C_t$ are connected. Therefore, the algorithm has computed a path $P$ that connects the components $C_s$ and $C_t$. Each vertex on that path corresponds to a component in $\mathcal{C}_0$, and each edge corresponds to an outgoing edge (discovered using the sketch information). Since $\mathcal{C}_0$ has at most $f+1$ components, $|P|\leq f+1$. 
Each such edge $e' \in P$ corresponds to an edge in $G$. Let $e_1=(x_1,y_1),\ldots, e_k=(x_k,y_k)$ be the $G$-edges corresponding to the edges of $P$ ordered from $C_s$ to $C_t$. Letting $y_0=s$ and $x_{k+1}=t$, we get that 
$y_i$ and $x_{i+1}$ belong to the same component in $\mathcal{C}_0$, for every $i \in \{0,\ldots, k\}$. 
The labeled path is given by $\widehat{P}=[s, x_1,y_1, x_2, y_2, \ldots y_k,t]$ where the edges $(y_i,x_{i+1})$ are labeled $1$ and the edges $(x_{i}, y_i)$ are labeled $0$. Each $0$-labeled edge is a real edge in $G$, and each $1$-labeled edge $(x_{i}, y_i)$ corresponds to a tree path $\pi(x_i, y_i)$ in $T \setminus F$. 
\end{proof}

\begin{figure}[h!]
\begin{center}
\includegraphics[scale=0.40]{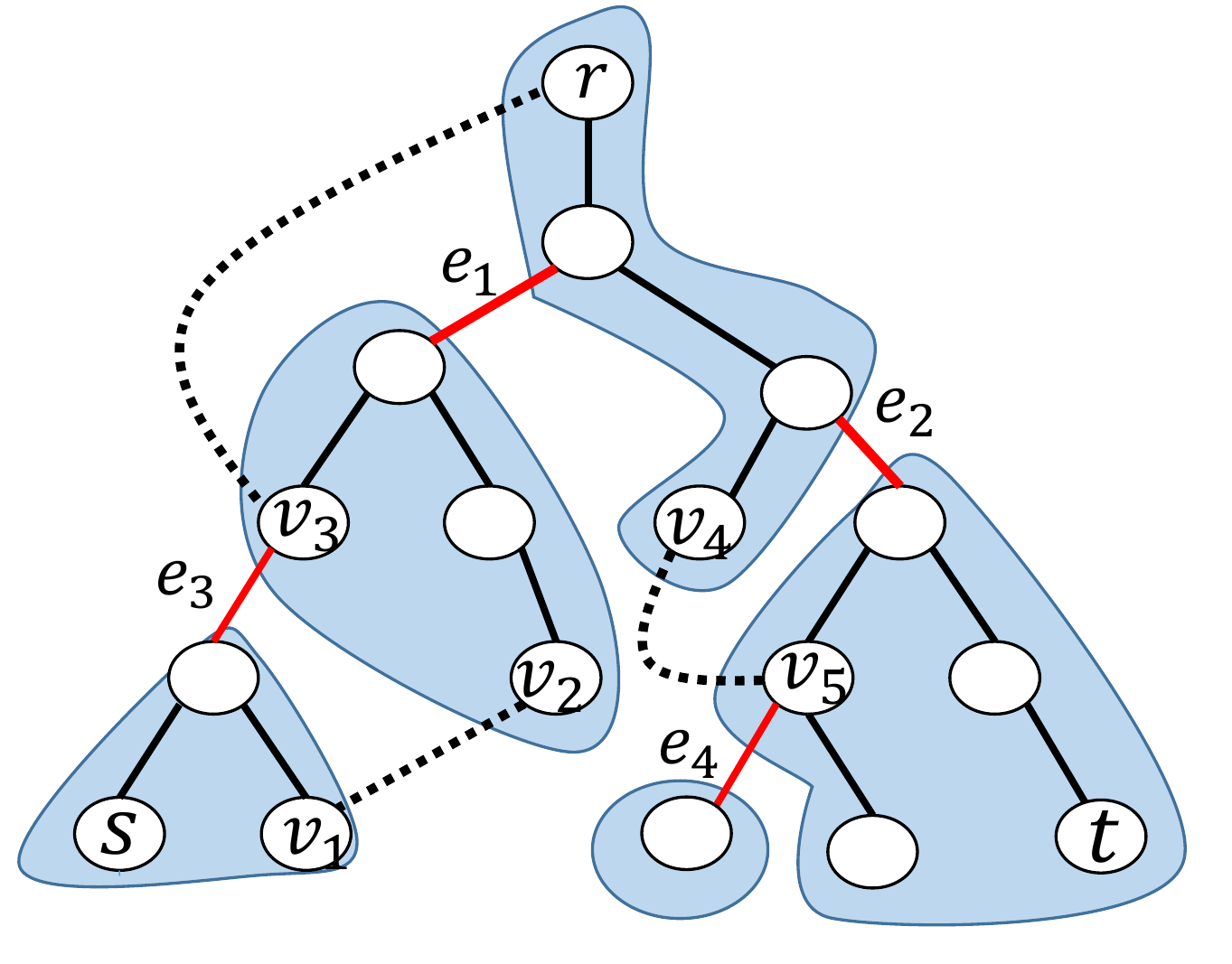}
\caption{\sf Shown is a tree $T$ with faulty edges $e_1,\ldots, e_4$. The $s$-$t$ path in $G \setminus F$ is represented by the path $\widehat{P}=[s,v_1]\circ (v_1,v_2) \circ [v_2,v_3] \circ (v_3,r) \circ [r,v_4] \circ (v_4, v_5) \circ [v_5,t]$. The recovery edges $(v_1,v_2), (v_3,r)$ and $(v_4, v_5)$ are shown in dashed lines.  \label{fig:succ-paths}
}
\end{center}
\end{figure}
\section{Fault-Tolerant Approximate Distance Labels}\label{sec:ft-distance}
Given integer parameters $f,k \geq 1$, an $(f,k)$ \emph{FT approximate distance labeling scheme} assigns labels $\FTDistLabel: V \cup E \to \{0,1\}^{q}$ such that given the labels of $s,t$ and a subset $F \subseteq E$, $|F|\leq f$, there exists a decoding algorithm that outputs a distance estimate $\delta_{G \setminus F}(s,t)$ satisfying:
$$\dist_{G\setminus F}(s,t) \leq \delta_{G \setminus F}(s,t) \leq k\cdot\dist_{G\setminus F}(s,t)~.$$

We next show that there is an efficient transformation from any FT connectivity labeling scheme into an FT approximate distance labeling scheme. This transformation increases the label size by a multiplicative factor of $\widetilde{O}(n^{1/k})$. This technique was first introduced by \cite{chechik2012f} in the context of distance sensitivity oracles, and it is based on the notion of tree covers.

\begin{definition}[Tree Covers]\label{def:tree-cover}
Let $G=(V,E)$ be an undirected graph with edge weights $\omega$, and let $\rho,k$ be two integers. Define $B_{\rho}(v)=\{ u \in V ~\mid~ \dist_G(u,v)\leq \rho\}$. A tree cover $\TreeCover(G, \omega, \rho,k)$ is a collection of rooted trees $\mathcal{T}=\{T_1,\ldots, T_\ell\}$ with root $r(T)$ for every $T \in \mathcal{T}$ such that:
\begin{enumerate}[noitemsep]
\item For every vertex $v$ there exists a tree $T \in \mathcal{T}$ such that $B_{\rho}(v) \subseteq T$.
\item The radius of each tree $T$ is at most $(2k-1)\cdot \rho$.
\item Each vertex participates in $(k \cdot n^{1/k})$ trees.
\end{enumerate}
Let $|\TreeCover(G, \omega, \rho,k)|$ denote the number of trees in the tree cover $\TreeCover(G, \omega, \rho,k)$.
\end{definition}

\begin{proposition}\cite{Peleg:2000}
For any $n$-vertex graph $G=(V,E, \omega)$, and any parameters $\rho,k$, one can compute tree covers $\TreeCover(G, \omega, \rho,k)$ in time $\widetilde{O}(|E(G)| \cdot n^{1/k})$.
\end{proposition}

\begin{lemma}[From Connectivity Labels to Approximate Distance Labels]\label{lem:reduction}
Let $G=(V,E, \omega)$ be a weighted undirected $n$-vertex graph where $\omega(e)\in [1,W]$, and let 
$\FTConnLabel: V \cup E \to \{0,1\}^s$ be an $f$-FT connectivity labeling scheme for $G$ with decoding time $t$. Then for every integer $k\geq 1$, there is an $(f,(8k-2)(|F|+1))$ FT approximate distance labeling scheme $\FTDistLabel: V \cup E \to \{0,1\}^{q}$ for $G$, where $q=O(s \cdot k \cdot n^{1/k}\cdot \log (nW))$, and with decoding time $\widetilde{O}(t \log{(nW)})$. 
\end{lemma}

\paragraph{The labeling algorithm.}
For every vertex $u$, the label $\FTDistLabel(u)$ consists of $K=\log(nW)$ sub-labels of FT connectivity labels in distinct subgraphs of $G$ defined as follows. The $i^{th}$ sub-label addresses all distances that are at most $2^i$ in $G$. Let $H_i$ be set of heavy edges in $G$ of weight at least $2^i$, and define the $i^{th}$ tree-cover by 
\begin{equation}\label{eq:TC-i}
\TreeCover_i=\TreeCover(G\setminus H_i,\omega, 2^i,k)~.
\end{equation}
For each tree $T_{i,j} \in \TreeCover_i$, the algorithm applies the FT connectivity scheme on the graph $G_{i,j}=G[V(T_{i,j})]$. For every vertex $u$ and $i \in \{1,\ldots, K\}$, let $i^*(u)$ be an index of a tree in $\TreeCover_i$ that covers the $2^i$-ball of $u$. I.e., $B_{2^i}(v) \subseteq T_{i,i^*(u)}$. 
The label of every $u \in V$ is then given by:
$$\FTDistLabel(u)=\{\langle \FTConnLabel_{G_{i,j},T_{i,j}}(u), i,j \rangle ~\mid~ i \in [1,K], j \in \{1,\ldots, |\TreeCover_i|\}, u \in G_{i,j}\} \bigcup \{i^*(u) ~\mid~ i \in [1,K]\}~.$$

Similarly, the label of each edge $e \in G$ contains the FT connectivity label of $e$ in each of the instances $(G_{i,j}, T_{i,j})$:
$$\FTDistLabel(e)=\{\langle \FTConnLabel_{G_{i,j},T_{i,j}}(e), i,j \rangle ~\mid~ i \in [1,K], j \in \{1,\ldots, |\TreeCover_i|\}, e \in G_{i,j}\}.$$ 

The time for assigning the labels is the time for constructing the tree cover and computing the indexes $i^{*}(v)$, and the time for assigning the connectivity labels on each one of the trees. The first part requires polynomial time. The second depends on the connectivity labels. For example, using our scheme from Section \ref{sec:ftconn-sketch} the time complexity of the second part is $\widetilde{O}(mn^{1/k})$, as it is linear in the total number of vertices and edges in the trees. 


\paragraph{The decoding algorithm.}
Consider the query $\langle s,t,F\rangle$. 
The algorithm has $K$ phases, in each phase $i \in [1,K]$ the decoding algorithm of the FT connectivity labels is applied on the instance $G_{i,i^*(s)}, T_{i,i^*(s)}$ where $G_{i,i^*(s)}$ contains the $2^i$ ball of $s$ in $G$. 
If $t \notin G_{i,i^*(s)}$, the phase $i$ ends and we continue to phase $i+1$. 
Otherwise, the algorithm decides if $s$ and $t$ are connected in $G_{i,i^*(s)} \setminus F$ in the following manner. 
Let $F_i=F \cap G_{i,i^*(s)}$, this subset of edges can be obtained from the labels of the $F$ edges. 
Since the labels of $s,t$ and $F_i$ contain the FT connectivity labels in the subgraph $G_{i,i^*(s)}$ and the tree $T_{i,i^*(s)}$, the algorithm can apply the decoding algorithm of the FT connectivity scheme.
If $s$ and $t$ are indeed connected in $G_{i,i^*(s)}\setminus F_i$, the algorithm returns the estimate $\delta_{G \setminus F}(s,t)= (4k-1) \cdot (|F|+1) \cdot 2^{i}$. Otherwise, it proceeds to the next phase.

Overall, let $i$ be the minimum index in $\{1,\ldots, K\}$ for which $s$ and $t$ are connected in the subgraph $G_{i,i^*(s)} \setminus F$. Then the decoding algorithm returns the distance estimate $\delta_{G \setminus F}(s,t)=(4k-1) \cdot (|F|+1) \cdot 2^{i}$. If no such $i$ exists, the decoding algorithm returns $\delta_{G \setminus F}(s,t)=\infty$, which implies that $s$ and $t$ are not connected in $G \setminus F$. 

The decoding time is $\widetilde{O}(t \log{(nW)})$, where $t$ is the decoding time of the connectivity labels, as we use the decoding algorithm of the connectivity labels $K$ times on the graphs $G_{i,i^*(s)}$. To obtain this, we need to make sure that given the labels of $s,t,F$ we can easily find their connectivity label in the graph $G_{i,i^*(s)}$ if exist. This can be easily done if we store the connectivity labels in a sorted order. 

\paragraph{Analysis.}
We now analyze the construction, and start by bounding the size of the labels. By the properties of the tree-cover in Def. \ref{def:tree-cover}, each vertex appears in $O(K \cdot k \cdot n^{1/k})$ subgraphs. Thus, $\FTDistLabel(u)$consists of $O(K \cdot k n^{1/k})$ FT connectivity labels and the label size is bounded by $O(K \cdot k n^{1/k} \cdot s)$ bits, as desired. Next, we show correctness. By the correctness of the FT connectivity labeling scheme, it is sufficient to show the following. Let $P_{s,t,F}$ be an $s$-$t$ shortest path in $G \setminus F$ of length $(2^{i-1}, 2^{i}]$. By the properties of the tree cover, there is a tree $T_{i,i^*(s)} \in \TreeCover_i$ that contains all the vertices of the path $P_{s,t,F}$. Therefore, we have that $s$ and $t$ are connected in $G_{i,i^*(s)}\setminus F$. Since the labels of $s, t$ and $F_i=F \cap G_{i,i^*(s)}$ contain the FT connectivity labels in $G_{i,i^*(s)}$, we get that the distance estimate returned by the algorithm satisfies that
$$\dist_{G \setminus F}(s,t)\leq \delta_{G \setminus F}(s,t) \leq (4k-1)(|F|+1)  \cdot 2^i \leq (8k-2)(|F|+1) \cdot \dist_{G \setminus F}(s,t)~.$$
To see this, let $j \leq i$ be the first index such that $s$ and $t$ are connected in $G_{j,j^*(s)}\setminus F$. The algorithm returns the estimate $(4k-1)(|F|+1) \cdot 2^j \leq (4k-1)(|F|+1) \cdot 2^i = (8k-2)(|F|+1) \cdot 2^{i-1} \leq (8k-2)(|F|+1) \cdot \dist_{G \setminus F}(s,t)$. To prove the left inequality, we show that if $s$ and $t$ are connected in $G_{j,j^*(s)}\setminus F$, there is indeed a path between them in $G \setminus F$ of length at most $\delta_{G \setminus F}(s,t) = (4k-1)(|F|+1) \cdot 2^j$. First, from the tree cover properties, the radius of $T_{j,j^*(s)}$ is at most $(2k-1)2^j$, implying that any two vertices in $T_{j,j^*(s)}$ are at distance at most $(4k -2) \cdot 2^j$ from each other. Now the graph $T_{j,j^*(s)} \setminus F$ has at most $|F|+1$ connected components. Since $G_{j,j^*(s)}\setminus F$ is connected, it implies that there is a path between $s$ and $t$ in $G_{j,j^*(s)}\setminus F$. This path traverses at most $|F|+1$ different components in $T_{j,j^*(s)} \setminus F$, and at most $|F|$ edges connecting them, each one of weight at most $2^j$. As the diameter of each component is bounded by $(4k - 2) \cdot 2^j$, the length of the path is at most $(4k - 2) \cdot 2^j \cdot (|F|+1) + 2^j \cdot |F| \leq (4k - 1) \cdot 2^j \cdot (|F| + 1)$, as needed.

\section{Compact Routing Schemes}
In this section, we explain how to use our FT distance labels to provide compact and low stretch routing schemes. This is the first scheme to provide an almost tight tradeoff between the space and the multiplicative stretch, for a constant number of faults $f=O(1)$.  Throughout this section, tree routing operations are performed by using the tree routing scheme of Thorup and Zwick \cite{thorup2001compact}.
\begin{fact}\label{fc:route-trees}[Routing on Trees]\cite{thorup2001compact}
For every $n$-vertex tree $T$, there exists a routing scheme that assigns each vertex $v \in V(T)$ a label $L_T(v)$ of $(1+o(1))\log n$ bits. Given the label of a source vertex
and the label of a destination, it is possible to compute, in constant time, the port number of the edge from the source that heads in the direction of the destination.
\end{fact}

We slightly modify the connectivity label of the edges and vertices by augmenting them with routing information. 
First, we augment the extended identifier of an edge (see Eq. (\ref{eq:extend-ID})) with port information and tree routing information, by having:
\begin{equation}\label{eq:edge-extended-routing}
\EID_T(e)=[\UID(e), \ID(u), \ID(v), \LCALabel_T(u), \LCALabel_T(v), \port(u,v), \port(v,u), L_T(u), L_T(v)]~,
\end{equation}
where $\port(u,v)$ is the port number of the edge $(u,v)$ for $u$, and the labels $L_T(u), L_T(v)$ are the tree routing labels taken from Fact \ref{fc:route-trees}. 
We then slightly modify the connectivity label of Eq. (\ref{eq:conn-vertex}) to include also the tree label $L_T(u)$from Fact \ref{fc:route-trees}, by defining 
\begin{equation}\label{eq:conn-vertex-label-routing}
\FTConnLabel_{G,T}(u)=\langle \LCALabel_T(u), \ID(u), L_T(u)\rangle~.
\end{equation}

%
%
%
%
%
Throughout this section, when applying the connectivity labels from Section \ref{sec:ftconn-sketch} on a graph $G$ with a spanning tree $T$, we use these modified extended identifiers and labels. This will also be the basis for the application of the distance labels of Section \ref{sec:ft-distance}. 
Similarly to the distance labels of Section \ref{sec:ft-distance}, we will apply the connectivity labels with respect to the different trees of the tree cover as discussed in Section \ref{sec:ft-distance}. 
Let $T_{i,j} \in \TreeCover_i$, recall that $G_{i,j}=G[V(T_{i,j})]$ and that $\mathcal{T}=\bigcup_{i=1}^K \TreeCover_i$ for $K=O(\log (nW))$. 

\begin{lemma}\label{lem:succint_path_routing}
Consider a triplet $s,t,F$ such that $s,t,F \in G_{i,j}$. \\
Given the connectivity labels $\{\FTConnLabel_{G_{i,j},T_{i,j}}(w)\}_{w \in F \cup \{s,t\}}$, we can determine w.h.p if $s$ and $t$ are connected in $G_{i,j} \setminus F$. If they are connected, we can output a labeled $s$-$t$ path $\widehat{P}$ of length $O(f)$ that provides a succinct description of the $s$-$t$ path in $G_{i,j} \setminus F$. The edges of $\widehat{P}$ are labeled by $0$ and $1$, where $0$-labeled edges correspond to $G_{i,j}$-edges and $1$-labeled edges $e=(x,y)$ correspond to $x$-$y$ paths in $T_{i,j} \setminus F$. For each $G_{i,j}$-edge, the succinct path description has the port information of the edge, and for each $x-y$ path, the description has the tree routing labels $L_{T_{i,j}}(x),L_{T_{i,j}}(y)$.
The length of the $s$-$t$ path encoded by $\widehat{P}$ is bounded by $(4k-1)(|F|+1)\cdot 2^i$. 
\end{lemma}

\begin{proof}
The proof follows the proof of Lemma \ref{lem:useful-recovery-edges}.
Using $\{\FTConnLabel_{G_{i,j},T_{i,j}}(w)\}_{w \in F \cup \{s,t\}}$, the decoding algorithm of Section \ref{sec:ftconn-sketch} determines if $s$ and $t$ are connected in $G_{i,j} \setminus F$. If they are connected, then from Lemma \ref{lem:useful-recovery-edges}, we get a succinct description of the $s$-$t$ path in $G_{i,j} \setminus F$. We next show that the algorithm indeed has the relevant port and tree routing information. For this note that all the vertices in the path $\widehat{P}$ obtained by Lemma \ref{lem:useful-recovery-edges} are either $s$ and $t$ or endpoints of the $|F|$ recovery edges found in the algorithm. The labels of $s$ and $t$ contain the tree routing information $L_{T_{i,j}}(s)$ and $L_{T_{i,j}}(t)$, and when the algorithm finds a recovery edge, it learns about its extended id $\EID_{T_{i,j}}(e)$ that has the port information and tree routing information of its endpoints. Any $G_{i,j}$-edge in $\widehat{P}$ is a recovery edge, hence the algorithm has its port information, and for any $x$-$y$ path in $T_{i,j} \setminus F$, the algorithm has the tree routing labels $L_{T_{i,j}}(x),L_{T_{i,j}}(y)$, as needed.
The stretch analysis follows the stretch analysis in Section \ref{sec:ft-distance}. It is based on the fact that $\widehat{P}$ has as most $|F|+1$ subpaths in $T_{i,j} \setminus F$, each of length at most $(4k-2)2^i$, and at most $|F|$ recovery edges of weight at most $2^i$.
\end{proof}

\subsection{Forbidden Set Routing (Faulty Edges are Known)}\label{sec:routing-known}
We start by describing the routing scheme in the forbidden set setting, where the faulty edges $F$ are known to the source vertex $s$. We show the following.

\begin{theorem}\label{thm:routing-known}[Forbidden-Set Routing]
For every integers $k,f$, there exists an $f$-sensitive compact routing scheme that given a message $M$ at the source vertex $s$, a label of the destination $t$, and labels of at most $f$ forbidden edges $F$ (known to $s$), routes $M$ from $s$ to $t$ in a distributed manner over a path of length at most $(8k-2)(|F|+1)\cdot \dist_{G \setminus F}(s,t)$. The table size of each vertex is bounded by $\widetilde{O}(n^{1/k} \log{(nW)})$. The header size of the messages is bounded by $\widetilde{O}(f)$ bits. The labels of vertices and edges have size $\widetilde{O}(n^{1/k} \log(nW))$.
\end{theorem}

\begin{proof}
The algorithm is based on the distance labels from Section \ref{sec:ft-distance} using the slightly modified connectivity labels (augmented with port and tree roting information). Recall that the distance labels are based on applying fault-tolerant connectivity labels on different graphs $G_{i,j}$, we use the slightly modified connectivity labels and the corresponding distance labels. 
The routing table of each vertex $u$ consists of its distance label $\FTDistLabel(u)$. The label of an edge $e$ is $\FTDistLabel(e)$. Each distance label has $\widetilde{O}(n^{1/k} \log(n W))$ bits. 

In the routing algorithm, the vertex $s$ is given the label $\FTDistLabel(t)$, and the labels $\{\FTDistLabel(e)\}_{e \in F}$, and it needs to route a message to $t$ in the graph $G \setminus F$. 
Recall that the algorithm from Section \ref{sec:ft-distance} works in $K$ phases, where in phase $i$ it checks if $s$ and $t$ are connected in the graph $G_{i,i^*(s)} \setminus F$ that contains the $2^i$-ball around $s$. Let $i$ be the first iteration where $s$ and $t$ are connected in $G_{i,i^*(s)} \setminus F$ according to the algorithm, and denote $G_i = G_{i,i^*(s)},T_i = T_{i,i^*(s)}$, and let $F_i = F \cap G_i$. The algorithm can also give a succinct description of an $s$-$t$ path in $G_i \setminus F_i$ following Lemma \ref{lem:succint_path_routing}. For this, note that we indeed have all the required information. The distance labels of edges in $F$ in particular contain the labels $\{\FTConnLabel_{G_i,T_i}(e)\}_{e \in F_i}$, and we can also tell which edges of $F$ are in $G_i$ from the labels. Also, the labels of $s,t$ contain the information $\ID_{T_i}(s),\ID_{T_i}(t)$ if they are both in $T_i$ (otherwise, they are not connected in level $i$).

The path $\widehat{P}$ as described in Lemma \ref{lem:succint_path_routing} is composed of $O(|F|)$ parts, where segment $(x,y)$ in the path corresponds either to an edge in $G_i$ or to a tree path in $T_i \setminus F$, it also has the relevant port and tree routing information. Our goal is to route a message according to this path. For this we add to the header of the message the description of $\widehat{P}$, the indexes $(i,i^*(s))$ of the tree we explore and an index $1 \leq q \leq 2|F|+1$ that represents the segment of $\widehat{P}$ we currently explore, initially $q=1$. Overall, the header size is $\widetilde{O}(f)$. To route a message according to the path, we work as follows. The header specifies the current segment in $\widehat{P}$. If the current segment corresponds to an edge $(x,y) \in G$, then $x$ uses the port information to route the message to $y$ and increases the index $q$. Otherwise, the current segment represents a tree path $(x,y) \in T_i$ and a vertex $u$ in this path uses its routing label in $T_i$ and the routing label of $y$ in $T_i$ (that is part of the header) to route the message towards $y$. When the message reaches $y$, it increases the index $q$. This completes the description of the routing process. The length of the path described is at most $(8k-2)(|F|+1)\cdot \dist_{G\setminus F}(s,t)$, as shown in Section  \ref{sec:ft-distance}. 
\end{proof}


\subsection{Fault-Tolerant Routing (Faulty Edges are Unknown)}\label{sec:route-unknown}
We now consider the more involved setting where the set of failed edges $F$ are unknown to $s$. In this case, an edge $(u,v) \in F$ is detected only when the message arrives, during the routing procedure, to one of the endpoints of $e$. Note that the routing scheme should, by definition, be prepared to any set of faulty edges $F$. However, the space bound of our scheme is required to be bounded by $\widetilde{O}(f n^{1+1/k})$, which is possibly much smaller than the number of graph edges $m$. This in particular implies that we cannot store the FT distance labels of all the graph edges. Nevertheless, we show that it is sufficient to explicitly store the labeling information for the tree edges in $\mathcal{T}=\bigcup_{i=1}^K \TreeCover_i$. The required information for the failed non-tree edges would be revealed throughout the process, by applying the decoding algorithm of Lemma \ref{lem:succint_path_routing}.
Our routing scheme eventually routes the message along the $s$-$t$ path encoded by the FT distance labels of $s,t$ and $F$. However, since the labels of $F$ are unknown in advance, the routing scheme will detect these edges in a trail and error fashion which induces an extra factor of $f$ in the final multiplicative stretch. This extra $f$ factor is also shown to be essential, in the end of the section.
We proceed by describing the routing tables.  

\paragraph{The routing labels and tables.} For ease of presentation, we first describe a solution with a multiplicative stretch of $O(kf^2)$, and \emph{global} space of $\widetilde{O}(f K \cdot n^{1+1/k})$, but the individual tables of some of the vertices might be large. We later on improve the space of each table to $\widetilde{O}(f^3 K \cdot n^{1/k})$ bits.

Recall that $\mathcal{T}=\bigcup_i^{K} \TreeCover_i$, for $K=O(\log (nW))$ is a collection of tree covers in all $K=\lceil \log (nW) \rceil$ distance scales, see Eq. (\ref{eq:TC-i}). For every vertex $v$, let $\deg_{\mathcal{T}}(v)=\sum_{T_{i,j} \in \mathcal{T}}\deg(u,T_{i,j})$ be the sum of degrees of $u$ in the collection of trees  $\mathcal{T}$. Recall that $G_{i,j}=G[V(T_{i,j})]$.
For the routing we apply the FT connectivity labels on the graphs $G_{i,j}$, similarly to Section \ref{sec:ft-distance}. 
\\ \\
\noindent \textbf{Routing labels.} The routing process uses at most $f'=f+1$ independent applications of randomized FT connectivity labels from Section 
\ref{sec:ftconn-sketch}, applied on each one of the graphs $G_{i,j}$. 
In more details, when we apply the labeling scheme on the graph $G_{i,j}$ with spanning tree $T_{i,j}$, we use $f'$ independent random seeds $\mathcal{S}_h$ to determine the randomness of the sketches. 
However, the seed $\mathcal{S}_{ID}$ used to determine the extended ids of edges in $G_{i,j}$ is fixed in the $f'$ applications, hence the extended identifiers of the edges (see Eq. (\ref{eq:extend-ID})) are fixed in all the $f'$ applications, and we only use fresh randomness to compute the sketch information using $f'$ independent seeds $\mathcal{S}^1_h,\ldots, \mathcal{S}^{f'}_h$. 
This process is done independently on each one of the graphs $G_{i,j}$. 

Denote the output connectivity labels obtained by the ${\ell}^{th}$ application of the scheme (using $\mathcal{S}^\ell_h$) on the graph $G_{i,j}$ by $\FTConnLabel^{\ell}_{G_{i,j},T_{i,j}}(w)$ for every $w \in E(G_{i,j})\cup V(G_{i,j})$. For every edge $e \in G_{i,j}$, define its $T_{i,j}$ routing label by
\begin{equation}\label{eq:route-edge-label}
    L_{route,i,j}(e)=
    \begin{cases}
      (\FTConnLabel^1_{G_{i,j},T_{i,j}}(e),\ldots,\FTConnLabel^{f'}_{G_{i,j},T_{i,j}}(e)),& \mbox{~for~} e \in T_{i,j} \\
     \EID_{T_{i,j}}(e),& e \in G_{i,j}\setminus E(T_{i,j})~.
    \end{cases}
\end{equation}
Every $L_{route,i,j}(e)$ label has $O(f \log^3 n)$ bits. 
In our routing algorithms, the $T_{i,j}$ routing labels of the discovered faulty edges will be added to the header for the message in order to guide the routing process. 
We now turn to define the routing labels of vertices. Recall that for a vertex $v$ and index $1 \leq i \leq K$, we denote by $i^*(v)$ an index such that the $2^i$-ball around $v$ is contained in $G_{i,i^*(v)}$.
The routing label $L_{route}(v)$ of $v$ 
For every \emph{vertex} $v$, the routing label of $v$ is given by
\begin{equation}\label{eq:Label-route-vertex}
L_{route}(v) = \{(i^*(v), \FTConnLabel^1_{G_{i,i^*(v)},T_{i,i^*(v)}}(v) | i \in [1,K]\}~.
\end{equation}

Note that by definition, the connectivity labels of the \emph{vertices} are the same in all $f'$ applications of the labeling algorithm, and therefore it is sufficient to include only one of these copies in the label. The size of the label is $O(K\log{n})=O(\log{n} \log{nW})$.
%
%
\\ \\
\noindent \textbf{Routing tables.} The routing table $R_{route}(v)$ of a vertex $v$ has the following information for every tree 
$T_{i,j}$ such that $v \in T_{i,j}$:
\begin{equation}\label{eq:route-table-ij}
R_{route,i,j}(v)=\{L_{route,i,j}(e), e \in E(v,T_{i,j})\} \cup \{\FTConnLabel^1_{G_{i,j},T_{i,j}}(v)\}~,
\end{equation}
where $E(v,T_{i,j})$ is the set of edges incident to $v$ in the tree $T_{i,j}$. 
The final routing table is given by $R_{route}(v)=\{R_{route,i,j}(v), (i,j) ~\mid~ T_{i,j} \in \mathcal{T}, v\in T_{i,j}\}$.

Since the connectivity labels are of size $\widetilde{O}(f)$, and as each $v$ appears in $\deg_{\mathcal{T}}(v)$ trees, the size of the table is $\widetilde{O}(f \deg_{\mathcal{T}}(v)).$ Since the total number of tree edges in $\mathcal{T}$ is bounded by $\widetilde{O}(K \cdot n^{1+1/k})$, this provides a global space bound of $\widetilde{O}(fK \cdot n^{1+1/k})$ bits.


\paragraph{The routing algorithm.} In the routing algorithm, the source vertex $s$ initially gets the routing label $L_{route}(t)$ (Eq. (\ref{eq:Label-route-vertex})) of the destination $t$ and its own routing table, $R_{route}(s)$, and its goal is to find the smallest radius graph $G_{i,j}$ such that $s$ and $t$ are connected in $G_{i,j} \setminus F$, and use it for routing. 
As the set $F$ is \emph{not} known in advance, the algorithm works in $K= O(\log{nW})$ phases, where in phase $i$ it tries to route a message in the graph $G_{i,i^*(t)}$ (which contains the entire $2^i$-radius ball of $t$). If $s$ and $t$ are connected in $G_{i,i^*(t)} \setminus F$ the algorithm succeeds, and otherwise we proceed to the next phase, corresponding to the distance scale of $2^{i+1}$. 
We next describe the algorithm for a single phase $i$, we denote $G_i = G_{i,i^*(t)}, T_i = T_{i,i^*(t)}$. Note that $s$ can deduce the index $i^*(t)$ from the routing label of $t$, and it can check if $s \in T_i$ using its routing table. If $s \not \in T_i$, we proceed to the next phase.

If $s \in T_i$, the routing procedure for phase $i$ has at most $|F|+1$ iterations. We maintain the following invariant in the beginning of each iteration $\ell \in \{1,\ldots, |F|+1\}$: (i) the iteration starts at vertex $s$, (ii) the algorithm has already detected a subset of $\ell-1$ faulty edges $F_\ell \subseteq F$, and (iii) the header contains the labels $\FTConnLabel_{G_i,T_i}(e)$ of all the edges $e \in F_\ell$. Each iteration $\ell \leq |F|+1$ terminates either at the destination vertex $t$, or at the source vertex $s$. In addition, w.h.p., if $s$ and $t$ are connected in $G_i \setminus F$, iteration $|F|+1$ terminates at $t$. The invariant holds vacuously for iteration $1$.

We now describe the $\ell^{th}$ iteration (of the $i^{th}$ phase) of the routing procedure given the invariant. The source vertex $s$ considers the $\ell^{th}$ copy of the FT connectivity labels, $\FTConnLabel^{\ell}_{G_i,T_i}(e)$ for every $e\in F_\ell$. 
Using the routing labels of the edges, that are part of the header, the routing label $L_{route}(t)$ (of Eq. (\ref{eq:Label-route-vertex})) and the routing table $R_{route}(s)$, $s$ can apply the decoding algorithm of 
Lemma \ref{lem:succint_path_routing} to determine if $s$ and $t$ are connected in $G_i \setminus F_\ell$. 
If the answer is no, the algorithm proceeds to the next phase $i+1$.
Otherwise, by applying the decoding algorithm of Lemma \ref{lem:succint_path_routing}, it computes the succinct path $\widehat{P}_\ell$. The path $\widehat{P}_\ell$ encodes an $s$-$t$ path in $G_i \setminus F_\ell$, that includes the relevant port and tree routing information of its vertices. The header of the message $H_\ell$ then contains 
$$H_\ell=\langle \widehat{P}_\ell, i, i^*(t), \{L_{route, i, i^*(t)}(e)\}_{e \in F_{\ell}}, q \rangle~,$$ where $q = O(f)$ is an index indicating the current segment of $\widehat{P}_\ell$ we explore. Note that the header $H_\ell$ contains the $f$ copies of connectivity labels of the $F_{\ell}$ edges, and not only the $\ell^{th}$ copy.
The size of the header is $\widetilde{O}(f^2)$, as the description of the path has size $\widetilde{O}(f)$, and additionally we have at most $f$ faulty edges with labels of size $\widetilde{O}(f)$.
%
Let $P_\ell$ be the $G$-path encoded by the path $\mathcal{P}_\ell$. The algorithm then routes the message along $P_\ell$ in the same manner as in Sec. \ref{sec:routing-known}. In the case where $P_{\ell}\cap F=\emptyset$, the iteration successfully terminates at the destination vertex $t$. From now on, we consider the case that $P_{\ell}$ contains at least one faulty edge. 

Let $e=(u,v)$ be the first edge (closest to $s$) on the path $P_\ell$ that belongs to $F$. Since $P_\ell \cap F_\ell=\emptyset$, it holds that $e \in F \setminus F_\ell$. Without loss of generality, assume that $u$ is closer to $s$ on $P_\ell$. Thus the faulty edge $e$ is detected upon arriving to the vertex $u$. 
In the case where $e$ is a \emph{non-tree edge}, then it must be a $G$-edge on $\widehat{P}_\ell$. Since this path has the extended ids $\EID_{T_i}(e)$ of its $G$-edges, and since the connectivity label of a non-tree edge $e$ is its extended identifier $\EID_{T_i}(e)$ in all the $f'$ applications of the scheme on $G_i$\footnote{This is because we use the same random seed $\mathcal{S}_{ID}$ in all these applications.}, $u$ can add $L_{route,i,i^*(t)}(e)=\EID_{T_i}(e)$ to the header of the message. Assume now that $e$ is a tree edge in $T_i$. The vertex $u$ then adds the routing label $L_{route,i,i^*(t)}(e)$ to the header of the message, as $e$ is a tree edge adjacent to $u$ it has this information in its routing table. Finally, it marks the header with the sign $R$, indicating that the message should now be routed in the reverse direction, until arriving $s$ again. This completes the description of iteration $\ell$. It is easy to see that the invariant is maintained. If $s$ and $t$ are connected in $G_i \setminus F$, after at most $f$ iterations all faulty edges are detected. In the last iteration, the path computed based on the labeling information is free from faulty edges, and the routing is completed (in the same manner as in Sec. \ref{sec:routing-known}) at the destination $t$. We next bound the multiplicative stretch of the routing.

\begin{claim}\label{cl:route-length}
Fix a set of faulty edges $F$, and let $s,t$ be vertices that are connected in $G \setminus F$. Then, the message is routed from $s$ to $t$ within $32k (|F|+1)^2 \cdot \dist_{G \setminus F}(s,t)$ steps, w.h.p.
\end{claim}
\begin{proof}
First note that since each iteration and each graph $G_{i}$ uses an independent set of FT connectivity labels, then in each phase and each iteration the decoding algorithm succeeds w.h.p. and outputs an $s$-$t$ path $\widehat{P}_\ell$ if exists. 

Assume that $\dist_{G \setminus F}(s,t) \in (2^{i-1},2^i]$. Then, $s$ and $t$ are connected in $G_i \setminus F$, as $T_i = T_{i,i^*(t)}$ contains the $2^i$-ball around $t$. We show that the algorithm terminates at $t$ in phase $i$ or before it, and that in any phase $j \leq i$, the routing algorithm traverses a path of length at most $2(4k-1)(|F|+1)^2 \cdot 2^j$.

Let $j \leq i$. In the $\ell$'th iteration of phase $j$, the algorithm first checks if $s$ and $t$ are connected in $G_j \setminus F_{\ell}$, where $F_{\ell}$ is the set of currently detected faults. If the answer is no, the algorithm proceeds to the next phase. Otherwise, it tries to route a message from $s$ to $t$ on the path encoded by $\widehat{P}_{\ell}$. The length of the path is bounded by $(4k-1)(|F|+1)\cdot 2^j$ from Lemma \ref{lem:succint_path_routing}. The algorithm either succeeds, or finds a faulty edge on the way in which case it returns to $s$ by traversing the same path on the reverse direction. Overall, the algorithm traverses a path of length at most $2(4k-1)(|F|+1)\cdot 2^j$, in this iteration. In all $|F|+1$ iterations of phase $j$, the length of the path explored is at most $2(4k-1)(|F|+1)^2 \cdot 2^j$. Summing over all iterations $j \leq i$, the stretch is bounded by $$\sum_{j=1}^{i} 2(4k-1)(|F|+1)^2 \cdot 2^j = 2(4k-1)(|F|+1)^2 \sum_{j=1}^i 2^j \leq 2^{i+2} (4k-1)(|F|+1)^2 \leq 32k (|F|+1)^2 \dist_{G \setminus F}(s,t).$$ The last inequality uses the fact that $2^{i-1} \leq \dist_{G \setminus F}(s,t).$ 

In the $i$'th phase, since $s$ and $t$ are connected in $G_i \setminus F$, then for any $F_{\ell} \subseteq F$, $s$ and $t$ are connected in $G_i \setminus F_{\ell}$, hence the algorithm always finds a path $\widehat{P}_{\ell}$. Hence, it either succeeds in routing the message to $t$ in one of the iterations (or one of the previous phases), or learns about all the failures $F$. In the latter case, in iteration $|F|+1$ it learns about a failure-free path $\widehat{P}_{|F|+1}$, and the routing terminates at $t$. This completes the proof.
\end{proof}
To conclude, we have the following.

\begin{theorem}
For every integers $k,f$, there exists an $f$-FT compact routing scheme that given a message $M$ at the source vertex $s$ and a label $L_{route}(t)$ of the destination $t$, in the presence of at most $f$ faulty edges $F$ (unknown to $s$) routes $M$ from $s$ to $t$ in a distributed manner over a path of length at most $32k (|F|+1)^2\cdot \dist_{G \setminus F}(s,t)$. 
The global table size is $\widetilde{O}(f \cdot n^{1+1/k} \log{(nW)})$.
The header size of the messages is bounded by $\widetilde{O}(f^2)$ bits, and the label size of vertices is $O(\log{(nW)} \log{n})$. 
\end{theorem}

\paragraph{Improving the size of the routing tables.} 
So far, we have described a routing scheme that consumes a total space of $\widetilde{O}(f\cdot n^{1+1/k}\log (nW))$ bits, and multiplicative stretch 
$32(|F|+1)^2 k$. We now explain the required modifications needed to providing routing tables with $\widetilde{O}(f^3\cdot n^{1/k})$ bits per vertex. The most space consuming information for a vertex $u$ is the connectivity labeling
information of the edges incident to $u$ in each of the trees $T_{i,j} \in \mathcal{T}$. As the degree of $u$ in some of the trees might be $\Theta(n)$, it leads to tables of possible super-linear size. To reduce the space of the individual tables, we apply a load balancing idea which distributes the labeling information incident to \emph{high}-degree vertices among their neighbors. 

Instead of storing the labeling information of $e=(u,v)$ at the routing tables of $u$ and $v$, we define 
for every tree $T \in \mathcal{T}$ and an edge $e=(u,v) \in T$, a subset $\Gamma_T(e)$ of vertices that store the connectivity labeling information of $e$ in $T$. We will make sure that the information on some vertex in $\Gamma_T(e)$ can be easily extracted in the routing procedure upon arriving one of its endpoints. In addition, we will make sure that each vertex stores the information only for a small number of edges in each of its trees. Consider an edge $e=(u,v)$ in a tree $T$, and assume, without loss of generality, that $u$ is the parent of $v$ in the tree $T$. In the case where $\deg(u,T)\leq f+1$, we simply let $\Gamma_T(e)=\{u,v\}$. That is, the label of $e$ is stored by both endpoints of $e$ (as before). The interesting case is where $\deg(u,T)\geq f+2$, in which case, $u$ might not be able to store the label of $e$, and will be assisted by its other children as follows.  Let $Child(u,T)=[v_1,\ldots, v_\ell]$ be the lexicographically ordered list of the children of $u$ in $T$.  The algorithm partitions $Child(u,T)$ into consecutive blocks of size $f+1$ (the last block might have $2f+1$ vertices). Letting $[v_{q,1}, \ldots, v_{q,f+1}] \subseteq Child(u,T)$ be the block containing $v$, define
$$\Gamma_T(e)=\{v_{q,1}, \ldots, v_{q,f+1}\}~.$$
Note that in particular, $v \in \Gamma_T(e)$. Thus, the label of $e$ is stored by $v$ and $\ell \in [f,2f-1]$ additional children of $u$ in $T$. 

We then modify the tree labels from Fact \ref{fc:route-trees} to contain the port information of $\Gamma_T(e)$. 
In order to do that, we will be using the more relaxed variant of Fact \ref{fc:route-trees}, we have:
\begin{claim}\label{cl:route-trees-port}
For every $n$-vertex tree $T$, there exists a (deterministic) routing scheme that assigns each vertex $v \in V(T)$ a label $L_T(v)$ of $O(f\log^2 n)$ bits and table $R_T(v)$ of $O(f\log n)$ bits. Given the label $L_T(t)$ of the target $t$  
and the routing table $R_T(u)$, the vertex $u$ can compute in $\widetilde{O}(f)$ time: (i) the port number of the edge $e=(u,v)$ on its tree path to $t$, and (ii) the port numbers of the neighbors of $u$ in the set $\Gamma_T(e=(u,v))$. 
\end{claim}
\begin{proof}
The proof follows by slightly modifying the simpler scheme of Fact \ref{fc:route-trees} by \cite{thorup2001compact}. Specifically, we will be using the routing scheme based on heavy-light tree decomposition. This scheme assigns each vertex $v$ labels of $O(\log^2 n)$ bits that contain the port information of the at most $O(\log n)$ light edges on the root to $v$ path in $T$.  The vertices are enumerated in DFS ordering, and the label of each vertex  contains its DFS range, and the specification of all light edges on its path in $T$ from the root, along with a port information of these edges.  The routing table of $v$ stores its DFS range, the port number of the (unique) heavy child of $v$ and also the port to its parent. In our modification, we augment the label of each vertex $u$ with the port information of $\Gamma_T(e')$ for every light edge $e'$ appearing on the root to $u$ path in $T$. Since there are $O(\log n)$ such light edges, the total label information is encoded in $O(f\log^2 n)$ bits. The routing table $R_T(u)$ is augmented with the port information for the set $\Gamma_T(e'')$, where $e''$ is the (unique) heavy child of $u$.  The routing scheme is then exactly as described at \cite{thorup2001compact}, only that in addition to the port of the next-hop $e=(u,v)$, we also obtain the port information of $\Gamma_T(e)$. This increases the labels and tables in the scheme of \cite{thorup2001compact} by a factor of $O(f)$, the claim follows.  
\end{proof}

Since the modified claim of tree routing defines now both tree routing labels and tables, we employ the following modifications. The extended identifier $\EID_T(e)$ of an edge $e=(u,v)$ from Eq. (\ref{eq:edge-extended-routing})
contains the modified tree labels and thus has $O(f\log^2 n)$ bits.
The \emph{routing labels} of Eq. (\ref{eq:route-edge-label}) are defined in the same manner only using the modified extended edge identifiers. The routing label of each edge has $\widetilde{O}(f^2)$ bits, and routing label of every vertex has $\widetilde{O}(f)$ bits. 
We are now ready to describe the more succinct \emph{routing tables} of each vertex $v$. We modify the definition of Eq. (\ref{eq:route-table-ij}) by letting:
\begin{equation*}\label{eq:route-table-ij-mod}
R_{route,i,j}(v)=\{L_{route,i,j}(e), e \in \Gamma_{T_{i,j}}(e)\} \cup \FTConnLabel^1_{G_{i,j},T_{i,j}}(v) \cup R_{T_{i,j}}(v)~,
\end{equation*}
thus the routing table $R_{route,i,j}(v)$ is augmented the tree routing tables $R_{T_{i,j}}(v)$ of Claim \ref{cl:route-trees-port}. In addition, $R_{route}(v)=\{R_{route,i,j}(v), (i,j) ~\mid~ T_{i,j} \in \mathcal{T}, v\in T_{i,j}\}$ as before.
We therefore have:
\begin{claim}\label{cl:route-balance-table}
The size of each routing table $R_{route}(v)$ is bounded by $\widetilde{O}(f^3 K n^{1/k})$ bits. 
\end{claim}
\begin{proof}
For every tree $T_{i,j}$ containing $v$, $v$ stores the routing labels for the tree $T_{i,j}$ of all edges in the set $E'(v,T_{i,j})=\{e \in T_{i,j} ~\mid~ v \in  \Gamma_{T_{i,j}}(e)\}$. Since each connectivity label of an edge contains the modified tree labels from Fact \ref{cl:route-trees-port}, it has $\widetilde{O}(f)$ bits, and as the routing label for $T_{i,j}$ contains $O(f)$ copies of this label, overall each routing label of an edge has $\widetilde{O}(f^2)$ bits. Observe that $|E'(v,T_{i,j})|=O(f)$ as each vertex stores the label of its parent in the tree, $O(f)$ child edges, and $O(f)$ child edges of its parent in the tree. Since each $v$ participates in $\widetilde{O}(K n^{1/k})$ trees, overall its routing table has $\widetilde{O}(f^3 K n^{1/k})$ bits, as required. 
\end{proof}

It remains to explain the required modifications for the routing procedure over a tree $T_i=T_{i,i^*(t)}$. Upon arriving to a vertex $u$ incident to a faulty \emph{tree} edge $e=(u,v)$ the procedure is as follows. If $e$ is a non-tree edge or if $u$ stores the connectivity label $\FTConnLabel_{G_i,T_i}(e)$\footnote{This covers the cases where $v$ is either a parent of $u$ or else, it is one of the at most $f+1$ children of $u$ in $T_i$.}, then $u$ adds the routing label of the edge to the header, as before. In the remaining case it must hold that $e$ is the edge incident to $u$ on its tree path to some vertex $y$. By using the tree routing scheme of Claim \ref{cl:route-trees-port} we have that given the tree routing labels  $L_{T_{i}}(u)$ and $L_{T_{i}}(y)$, the vertex $u$ can also obtain the port numbers of its $\ell \in [f,2f-1]$ children in $\Gamma_{T_{i}}(e)$ that store the label $\FTConnLabel_{G_i,T_i}(e)$. Since there are at most $f$ edge faults in the network, and $\Gamma_{T_{i,j}}(e)$ contains information on at least $f+1$ ports of $u$'s neighbors that contain the label of $e$, the vertex $u$ can access a non-faulty neighbor, say $w$, that has the label information of $e$. That vertex can then add the labeling information of $e$ to the header of the message, and the routing algorithm proceeds as before. Since we use the modified tree labels of Claim \ref{cl:route-trees-port}, each connectivity label has $\widetilde{O}(f)$ bits, and each routing label of an edge for a tree $T_{i,j}$ has $\widetilde{O}(f^2)$ bits. Since the header stores the routing labels of $O(f)$ edges, it consists of $\widetilde{O}(f^3)$ bits.

The stretch is still bounded by $32k (|F|+1)^2\cdot \dist_{G \setminus F}(s,t)$, as we next explain. Recall that in the proof of Claim \ref{cl:route-length}, we bounded the length of the path we explore in one iteration of the algorithm of phase $j$ by $2(4k-1)(|F|+1)2^j.$ In the new scheme, when we discover a faulty edge, the vertex $u$ may send messages to $|F|+1$ neighbors until it finds the label of the edge. This adds at most $2(|F|+1)2^j$ to the stretch, as the weight of edges in the tree of phase $j$ is at most $2^j$, and we may send messages in both directions. This gives that the length of the path we explore in one iteration is now at most $2(4k-1)(|F|+1)2^j+2(|F|+1)2^j=8k(|F|+1)2^j.$ The rest of the analysis proceeds as in the proof of Claim \ref{cl:route-length}, and gives that the stretch is bounded by $32k (|F|+1)^2\cdot \dist_{G \setminus F}(s,t)$ (we get the same bound as in the original proof we bounded $2(4k-1)$ with $8k$ during the analysis).
We therefore have:
\begin{theorem}\label{thm:routing-unknown}[Fault-Tolerant Routing]
For every integers $k,f$, there exists an $f$-sensitive compact routing scheme that given a message $M$ at the source vertex $s$ and a label $L_{route}(t)$ of the destination $t$, in the presence of at most $f$ faulty edges $F$ (unknown to $s$) routes $M$ from $s$ to $t$ in a distributed manner over a path of length at most $32k (|F|+1)^2\cdot \dist_{G \setminus F}(s,t)$. The routing labels have $\widetilde{O}(f)$ bits, the table size of each vertex is $\widetilde{O}(f^3 \cdot n^{1/k} \log(nW))$. The header size of the messages is bounded by $\widetilde{O}(f^3)$ bits. 
\end{theorem}

\paragraph{Lower Bound.} Finally, we show that the price of not knowing the set of faulty edges $F$ in advance might indeed incur a multiplicative stretch of $\Omega(f)$. 

\begin{proof}[Proof of Theorem \ref{thm:lb-routing}]
Consider a graph that consists of $f+1$ vertex disjoint $s$-$t$ paths, each of length $L=\Theta(n/f)$. The last edge of each of the paths, except for one, is faulty. Assume that the non-faulty path is chosen uniformly at random. Since the routing scheme is oblivious to the faulty edges, it can discover a faulty edge only upon sending the message to one of the edge endpoints. The expected length of the routing is given by:
$$\frac{L}{f+1} +2L \cdot \left(1-\frac{1}{f+1} \right)\cdot \frac{1}{f} + \ldots+ \left(f+1\right)L\cdot \prod_{i=0}^{f-1} \left(1-\frac{1}{f+1-i}\right)=\Omega(f L)~.$$ 
Since the $s$-$t$ shortest path under these faults is $L$, the proof follows. See Fig. \ref{fig:LB-stretch} for an illustration.
\end{proof}

\begin{figure}[h!]
\begin{center}
\includegraphics[scale=0.40]{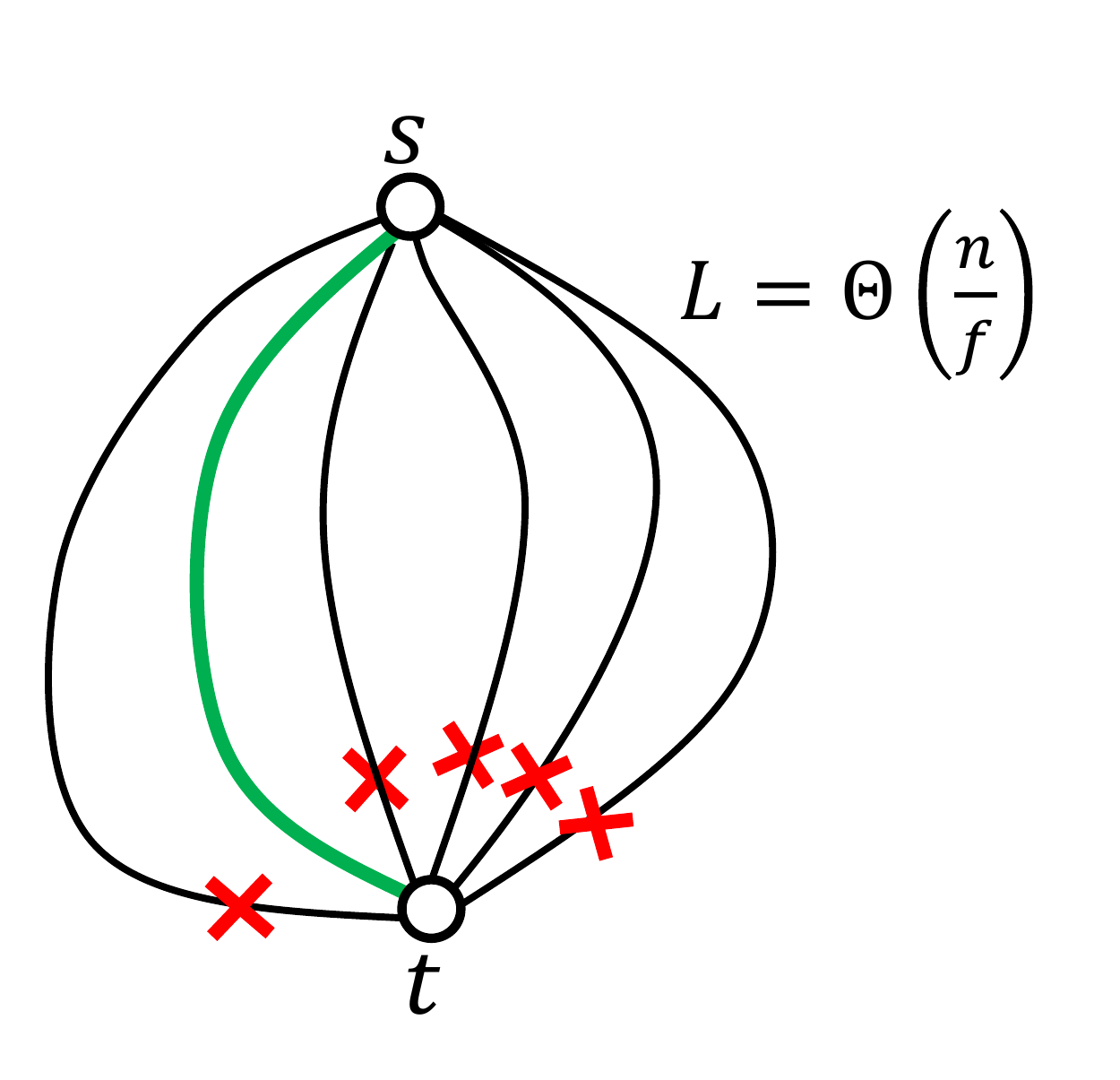}
\caption{\sf Illustration for a stretch lower bound for any FT routing schemes. The $s$-$t$ pair are connected by $f+1$ vertex disjoint paths of length $L$. Since the faulty-edge is the last edge of the path, the routing requires $\Omega(L)$ steps to discover a single faulty edge. As the non-faulty path is chosen uniformly at random, in expectation, the routing requires $\Omega(fL)$ steps.  \label{fig:LB-stretch}
}
\end{center}
\end{figure}


\bibliographystyle{alpha}
\bibliography{crypto}

\appendix
\section{Additional Definitions}

\begin{definition}[Pairwise Independence Hash Functions]\label{def:pairwise}
Let $\mathcal{H}$ be a family of functions from $\{1,\ldots, N\}$ to $\{1,\ldots, M\}$. The family $\mathcal{H}$ is \emph{pairwise independent} if for every $x,y \in \{1,\ldots, N\}$ such that $x \neq y$ and for every $a,b \in \{1,\ldots, M\}$ it holds that 
$$\Pr_{h \in \mathcal{H}}[h(x)=a \wedge h(y)=b]=1/M^2~.$$
That is, if $h$ is chosen uniformly at random from $\mathcal{H}$, then the random variable $h(x)$ and $h(y)$ are uniformly distributed and pairwise independent. 
\end{definition}

\begin{fact}\label{fc:pairwise}\cite{TCS-010}
There is an explicit family $\mathcal{H}$ of pairwise independent has functions from $\{0,1\}^n \to \{0,1\}^m$ constructed using $O(\max\{
m,n\})$ bits and computable in $\poly(n,m)$ time.
\end{fact}
\section{Overview of the Cycle Space Sampling Technique} \label{sec:cycle_space_overview}

The cycle space sampling technique allows to detect cuts in a graph using a connection between cuts and cycles in a graph.
This beautiful technique was introduced by Pritchard and Thurimella \cite{pritchard2011fast}, that showed its applicability for distributed algorithms identifying small cuts in a graph. We next give a short overview of the technique, for full details see \cite{pritchard2011fast}.

The \emph{cycle space} of a graph is the family of all subsets of edges $F$ that have even degree at each vertex, any such subset $\phi \subseteq E$ is called a \emph{binary circulation}. The \emph{cut space} is the family of all induced edge cuts. It is easy to see that if we take a cycle $C$ in a graph and an induced edge cut, then the number of edges of the cycle that cross the cut is even. The cycle space technique extends this observation and shows that the cycle space and cut space are orthogonal vector spaces. Using this, they show the following (see Propositions 2.2 and 2.5 in \cite{pritchard2011fast}).

\begin{claim} \label{claim_cycle} 
Let $\phi$ be a uniformly random binary circulation and $F \subseteq E$. Then
$$Pr[|F \cap \phi| \ is \ even] = \left\{
                \begin{array}{ll}
                  1,\ if\ F\ is\ an\ induced\ edge\ cut\\
                  1/2,\ otherwise
                \end{array}
              \right. $$ 
\end{claim}   

Hence, sampling a random binary circulation allows to detect if a subset of edges is an induced edge cut with probability $1/2$. To reduce the failure probability to $1/2^b$ we can choose $b$ random binary circulations. To use this technique, the authors provide an efficient way to sample a random binary circulation, we describe next. Let $T$ be a spanning tree of the graph. For any non-tree edge $e$, adding $e$ to the graph creates a cycle. These cycles are the \emph{fundamental cycles}, and it is shown that the \emph{fundamental cycles} are a basis for the cycle space. Based on this, they show that sampling a random binary circulation can be done by choosing each fundamental cycle with probability $1/2$, or equivalently choosing each non-tree edge with probability $1/2$. The binary circulation $\phi$ sampled has all the non-tree edges sampled, and each tree edge that appears in odd number of sampled cycles. Given the sampled non-tree edges in $\phi$, the tree edges in $\phi$ can be identified using a simple scan of the tree, as shown in \cite{pritchard2011fast}. Choosing $b$ random binary circulations, is equivalent to choosing a $b$-bit random string $\phi(e)$ for each non-tree edge. For a tree edge $t$, we define $\phi(t) = \oplus_{e \in C_t} \phi(e)$, where $C_t$ are all non-tree edges $e$ such that $t$ is in the fundamental cycle of $e$. This again can be computed by a simple scan of the tree, and takes $O((n+m)b)$ time if the labels have size $b$. This gives the following.

\cycle*

To see this, let $\phi_1,...,\phi_b$ be the sampled binary circulations. If $F$ is an induced edge cut, then from Claim \ref{claim_cycle}, for every sampled circulation $\phi_i$, we have that $|F \cap \phi_i|$ is even, and hence for all $i$, the $i$'th bit of $\Moplus_{e \in F} \phi(e)$ is equal to 0 as needed. Otherwise, for all $i$, the $i$'th bit $\Moplus_{e \in F} \phi(e)$ equals $0$ with probability $1/2$, hence the probability that the whole vector equals $0$ is $1/2^b$, as needed.
\section{Missing Proofs}\label{sec:miss-proof}

\APPENDUNIQUEID

\APPENDLEMMUNIQUE

\APPENDLABELCONSISE

\APPENDSKETCHPROP

\APPENDCONNLABELSKETCH

\end{document}